\documentclass[10pt,a4paper]{article}
\usepackage[T1]{fontenc}
\usepackage[utf8]{inputenc}
\usepackage[english]{babel}
\usepackage{a4wide}
\usepackage{authblk}
\usepackage{amsfonts}
\usepackage{xspace}
\usepackage{nameref,hyperref}
\usepackage{stmaryrd}
\usepackage{mathrsfs}
\usepackage{float}
\usepackage{graphicx}
\usepackage{times}
\usepackage{url}
\usepackage{amsmath,amsthm,amssymb,color}
\usepackage{multicol}
\usepackage{caption}
\usepackage{subcaption}
\usepackage{numprint}
\usepackage{cleveref}
\usepackage{tikz}
\usepackage{physics}
\usepackage{pdflscape}
\usepackage{array,multirow}
\usepackage{relsize}
\usepackage{pxfonts}
\usepackage{markdown}
\usepackage{empheq}
\usepackage{verbatim}
\usepackage{multirow}
\usepackage{algorithm}
\usepackage{algpseudocode}
\usepackage[normalem]{ulem} 
\usepackage{xcolor,colortbl}
\usepackage[backend=bibtex,firstinits,style=alphabetic,url=false,maxcitenames=3,maxnames=8]{biblatex}

\DeclareNameAlias{labelname}{given-inits}

\algnewcommand\algorithmicforeach{\textbf{for each}}
\algdef{S}[FOR]{ForEach}[1]{\algorithmicforeach\ #1\ \algorithmicdo}
\algnewcommand{\IIf}[1]{\State\algorithmicif\ #1\ \algorithmicthen}
\algnewcommand{\ElseIIf}[1]{\State \textbf{else} \algorithmicif\ #1\ \algorithmicthen}
\algnewcommand{\ITE}[3]{\State\algorithmicif\ #1\ \algorithmicthen\ #2\ \algorithmicelse\ #3\ \algorithmicend\ \algorithmicif}
\algnewcommand{\EndIIf}{\unskip\ \algorithmicend\ \algorithmicif}
\algnewcommand{\IFor}[2]{\State\algorithmicfor\ #1\ \algorithmicdo \ #2\ \algorithmicend\ \algorithmicfor}

\usepackage{mathtools}

\newcommand{\np}[1]{\numprint{#1}}
\renewcommand{\innerproduct}[2]{\left\langle #1\mid#2 \right\rangle}
\newcommand{\ctimes}{\cdot}

\newcommand{\vqa}[2]{\ensuremath{#1; #2}}

\theoremstyle{plain}
\newtheorem{theorem}{Theorem}
\newtheorem{corollary}[theorem]{Corollary}
\newtheorem{lemma}[theorem]{Lemma}
\newtheorem{proposition}[theorem]{Proposition}

\theoremstyle{definition}
\newtheorem{definition}[theorem]{Definition}

\theoremstyle{definition}
\newtheorem{example}{Example}
\colorlet{cRed}{red!100!}

\begin{document}

\title{An abstract structure determines the contextuality degree of observable-based Kochen-Specker proofs\footnote{Author version
of~\cite{mg25}.}}

\author{Axel Muller}
\author{Alain Giorgetti\footnote{Corresponding author, \texttt{alain.giorgetti@femto-st.fr}}}

\affil{Université Marie et Louis Pasteur, CNRS, institut FEMTO-ST,\\
 F-25000 Besançon, France}

\date{}

\maketitle

\begin{abstract}
This article delves into the concept of quantum contextuality, specifically focusing on proofs of the Kochen-Specker theorem obtained by assigning Pauli observables to hypergraph vertices satisfying a given commutation relation. The abstract structure composed of this hypergraph and the graph of anticommutations is named a hypergram. Its labelings with Pauli observables generalize the well-known magic sets. A first result is that all these correct quantum labelings of a given hypergram inherently possess the same degree of contextuality. Then we provide a necessary and sufficient condition for the existence of such quantum labelings and an efficient algorithm to find one of them. We finally attach to each assignable hypergram an abstract notion of contextuality degree. By presenting the study of observable-based Kochen-Specker proofs from the perspectives of graphs and matrices, this abstraction opens the way to new methods to search for original contextual configurations.
\end{abstract}

\section{Introduction}
\label{seCintroduction}

In classical physical theories, the measured value of a physical quantity does
not depend on that of other quantities simultaneously measured, called its
\emph{context}. This independence no longer holds in quantum theory, where
Kochen-Specker theorem predicts the existence of experiments whose measurement
outcomes necessarily depend on measurements that are simultaneously measurable
in principle. This phenomenon is called \emph{quantum contextuality} (see,
e.\,g., \cite{bcgkl} for a recent comprehensive review of this topic). It is a
core aspect of quantum mechanics, especially for quantum computation.

Several approaches have been proposed to study quantum contextuality, such as
binary~\cite{Cleve2014} and linear~\cite{CLS17} constraint systems, which are
generalizations of the Boolean case~\cite{FOC25}, and graph-based~\cite{CSW14}
and hypergraph-based~\cite{AFL+15} approaches, where vertices are events and
edges are mutually exclusive events. Note that those hypergraphs are not the
same as the hypergraphs we introduce in this work, whose vertices represent
observables and whose hyperedges represent contexts. There is also a homotopical
approach~\cite{OR20} that describes topological criteria for the commutation
relations of quantum observables, and a study of contextuality in the framework
of Lie algebras~\cite{ACE+25}.

This work is about (observable-based) \emph{contextuality proofs}, whose
measurements are multi-qubit Pauli observables. These proofs are
\emph{state-independent}, because their measurements reveal quantum
contextuality when applied to any initial quantum state. When the number of
qubits of the Pauli observables is small enough, they are \emph{testable}, in
the sense that they can be turned into experimental tests of contextuality on
existing quantum computers~(see, e.\,g., \cite{KZGKGCBR09} or \cite{Hol21}).

Structurally, these contextuality proofs are described by hypergraphs whose
vertices are multi-qubit Pauli observables and whose hyperedges, also called
\emph{contexts}, group together compatible, that is, commuting observables whose
product is either the identity matrix (\emph{positive} hyperedge or context) or
its opposite (\emph{negative} hyperedge or context). How much a proof is
contextual can be quantified by an integer called its \emph{contextuality
degree}~\cite{DHGMS22}, which is the minimal number of context signs which
cannot be satisfied by any assignment of all its observables/vertices by one of
their measurement values (see~\Cref{contextualitySec} for a detailed
mathematical definition).

A widely studied subfamily of contextuality proofs is that of \emph{contextual
configurations}~\cite{HS17}, aka. \emph{magic sets}, whose observables belong to
an even number of contexts (\emph{parity condition}), whose number of negative
contexts is odd (\emph{oddness condition}), and which are incidence geometries,
meaning that two observables share at most one context (\emph{incidence
condition}). The contextuality of a magic set is an immediate logical
consequence of these conditions~\cite{HS17}. Typical examples are the
Mermin-Peres squares~\cite{mermin,peres}, composed of nine two-qubit observables
and six contexts of three observables, with one or three negative contexts among
them.

Previous works inspired by finite geometry~\cite{SdHG21,MSGDH22,MSGDH24} exhibit
contextuality proofs which do not satisfy the first two conditions of magic
sets. For example, multi-qubit doilies~\cite{MSGDH22} comprise three contexts
per observable, thus do not satisfy the parity condition. Moreover, three-qubit
doilies with four negative lines do not satisfy the oddness condition, but also
provide contextuality proofs. Actually, whatever their number of qubits and
configuration of negative contexts, all doilies have been proved to be
contextual and to admit the same contextuality degree, whose value
is~3~\cite[Proposition~1]{MSGDH22}. Independently, the non-contextual bound of
magic sets -- linearly related to their contextuality degree, as detailed
in~\Cref{contextualitySec} -- has recently been shown not to depend on their
number of qubits~\cite[Theorem~2]{TLC22}.

At first glance, according to its definition, the degree of a contextuality
proof depends on the distribution of positive and negative contexts in it,
itself arising from the Pauli observables which label their vertices. The main
objective of this work is to clarify this dependence. We achieve this by
introducing the notion of hypergram (\Cref{hypergramDef}), which is an
observable-free graph- and hypergraph-based structure admitting a definition of
contextuality degree. It is \emph{abstract} because it is defined independently
of a number of qubits and more generally without recourse to any quantum-related
concept.

In this paper, we first define hypergrams and related notions (\Cref{seCdefs}).
Then, we demonstrate that the contextuality degree of its labelings by Pauli
observables (hereafter called ``Pauli assignments'') does not depend on their
number of qubits (\Cref{seCproof}). Then we derive from~\cite{TLC22v2} (which
provides supplemental material for~\cite{TLC22}) a necessary and sufficient
condition for a hypergram to admit a Pauli assignment and we propose an
efficient algorithm to find such an assignment when this condition holds
(\Cref{seCcns}). An immediate consequence is that only one adequate labeling of
vertices with Pauli observables is sufficient to compute the degree.
In~\Cref{msgdh23Sec} we present several examples of hypergrams, labeled with
minimal numbers of qubits and negative contexts. A comparison with related work
is provided in~\Cref{seCcompare}.

\section{Definitions and notations}
\label{seCdefs}

After \Cref{seCbackground,symplSec} providing minimal essential background about
the Pauli group and its relation with symplectic polar spaces, \Cref{ccsSec}
introduces a new abstract structure, composed of a hypergraph and a graph, which
will later be shown to admit a notion of contextuality degree, inherited from
the contextuality degree common to all the labelings of its vertices by Pauli
observables. The remainder of the section brings together definitions from
independent previous work, mainly~\cite{TLC22v2} and~\cite{DHGMS22,MSGDH24}, and
exhibits correspondences between these definitions, when it is useful, for
example between the notions of ``Pauli assignment'' and ``quantum
configuration'' in~\Cref{assignQCsec}, and between the notions of
``contextuality degree'' and ``noncontextual bound'' in~\Cref{contextualitySec}.

\subsection{Multi-qubit Pauli group}
\label{seCbackground}

Let
\begin{equation*}
X = \left(
\begin{array}{rr}
0 & 1 \\
1 & 0 \\
\end{array}
\right),~~
Y = \left(
\begin{array}{rr}
0 & -\text{i} \\
\text{i} & 0 \\
\end{array}
\right)~~{\rm and}~~
Z = \left(
\begin{array}{rr}
1 & 0 \\
0 & -1 \\
\end{array}
\right)
\label{paulis}
\end{equation*}
be the Pauli matrices, $I$ the $2\times 2$ identity matrix, `$\otimes$' denote
the tensor product of matrices and $I^{\otimes n}$ denote the $n$-fold tensor
$I\otimes I \otimes\ldots\otimes I$ of the identity. A local \emph{$n$-qubit
(Pauli) observable} is a tensor product $G_1 \otimes G_2 \otimes \cdots \otimes
G_n$ with $G_i \in \{I,X,Y,Z\}$, usually denoted $G_1 G_2 \cdots G_n$, by
omitting the symbol $\otimes$ for the tensor product. Let '$\ctimes$' denote the
matrix product and $M^2$ denote $M \ctimes M$. It is easy to check that $X^2 =
Y^2 = Z^2 = I$, $X \ctimes Y = \text{i}Z = -Y \ctimes X$, $Y \ctimes Z =
\text{i}X = -Z \ctimes Y$, and $Z \ctimes X = \text{i}Y = -X \ctimes Z$. The
$n$-qubit observables with the multiplicative factors $\pm 1$ and $\pm
\text{i}$, called \emph{phase}, form the (\emph{generalized}) ($n$\emph{-qubit})
\emph{Pauli group} $\mathcal{P}^{\otimes n} =
(\{1,-1,\text{i},-\text{i}\}\times\{I,X,Y,Z\}^{\otimes n},\ctimes)$.

\subsection{Connection with symplectic polar spaces}
\label{symplSec}

Let $a$ and $b$ be two elements of the two-element field $\mathbb{F}_2 =
\{0,1\}$. Their sum, denoted $a+b$, and their product, denoted $ab$,
respectively correspond to the logical operations of exclusive disjunction and
conjunction, when $0$ encodes ``false'' and $1$ encodes ``true''.

The $2n$-dimensional vector space $\mathbb{F}_2^{2n}$ over $\mathbb{F}_2$ has
vector subspaces for each dimension $0 \leq k \leq 2n$. A subspace is
\emph{totally isotropic} if any two vectors $x$ and $y$ in it are mutually
orthogonal ($\innerproduct{x}{y} = 0$), for the symplectic form
$\innerproduct{.}{.}$ defined by
\begin{equation} 
\innerproduct{x}{y} = x_1y_2 + x_2y_1 + x_3y_4 + x_4y_3 + \dots 
              + x_{2n-1} y_{2n} + x_{2n} y_{2n-1}.
\label{symplf}
\end{equation}

The totally isotropic subspaces of $\mathbb{F}_2^{2n}$, without their zero
vector, form the \emph{symplectic (polar) space} $\mathcal{W}(2n-1,2)$
 of projective dimension $2n-1$. This name, in which $2$
is the order of the field $\mathbb{F}_2$, is hereafter shortened as $W_n$. In
other words, a (totally isotropic) subspace of $W_n$ of (projective) dimension
$k$, with $1 \leq k \leq n-1$, is a totally isotropic vector subspace of
$\mathbb{F}_2^{2n}$ of dimension $k+1$ without its $0$.

The $4^n-1$ phase-free $n$-qubit observables $G_1 \cdots G_j \cdots G_n$ in
$\mathcal{P}^{\otimes n}$ other than the identity $I^{\otimes n}$ are
bijectively identified with the $4^n-1$ vectors
$(x_1,x_2,\ldots,x_{2j-1},x_{2j},\ldots,x_{2n-1},x_{2n})$ which are the points
of $W_n$, by the extension $\psi:
\{I,X,Y,Z\}^{\otimes n}~\rightarrow~\mathbb{F}_2^{2n}$ of the encoding bijection
$\psi : \{I,X,Y,Z\} \rightarrow \mathbb{F}_2^{2}$ defined by
\begin{equation}
\psi(I) = (0,0),~\psi(X) = (0,1),~\psi(Y) = (1,1) {\rm~and}~\psi(Z) = (1,0).
\label{paulipts}
\end{equation}
This extension is defined by $\psi(G_1 \cdots G_j \cdots G_n) =
(x_1,x_2,\ldots,x_{2j-1},x_{2j},\ldots,x_{2n-1},x_{2n})$ with $\psi(G_j) =
(x_{2j-1},x_{2j})$ for $1~{\leq}~j~{\leq}~n$.

With the symplectic form defined by~(\ref{symplf}), two commuting observables
are represented by two orthogonal vectors.

\subsection{Abstract structure}
\label{ccsSec}

A \emph{simple graph} is an undirected graph without multiple edges and
\emph{loops}, i.\,e., edges $\{v,v\}$ for some vertex $v$. A \emph{hypergraph}
$\mathcal{H} = (V,H)$ is a finite set $V$ of \emph{vertices} and a (finite) set
$H$ of \emph{hyperedges}, which are (distinct) subsets of vertices in $V$. Two
vertices are \emph{adjacent} (in $\mathcal{H}$) if they are in the same
hyperedge of $H$. The \emph{complement graph} of the hypergraph $\mathcal{H} =
(V,H)$ is the (simple) graph $\text{cplt}(\mathcal{H}) = (V,\text{cplt}(H))$
with the same vertices as $\mathcal{H}$ and whose (undirected and non-loop)
edges are the sets of two distinct non-adjacent vertices in $\mathcal{H}$.
Formally, $\{v,v'\} \in \text{cplt}(H) \Leftrightarrow v \neq v' \land
\nexists\, h \in H.\; \{v,v'\} \subseteq h$. Following~\cite{GR01}, a graph is
said to be \emph{reduced} if it has no isolated vertex and no pair of vertices
with the same \emph{neighborhood}, which is their set of adjacent vertices.

With these definitions in mind, we can now introduce the hypergram that we
propose as the abstract structure underlying operator-based contextuality proofs
and determining their degree of contextuality.

\begin{definition}\label{hypergramDef} A \emph{hypergram} is a triple $(V,H,G)$
where $V$ is a non-empty finite set of \emph{vertices}, $(V,H)$ is a
hypergraph (called \emph{context hypergraph}) without isolated vertices (outside
any hyperedge) and empty hyperedges and $(V,G)$ is a simple reduced graph
(called \emph{anticommutation graph}) such that $G \subseteq \text{cplt}(H)$. We
say that two vertices $i$ and $j$ \emph{commute} if $\{i,j\} \not\in G$.
\end{definition}

By definition, each vertex commutes with itself. The inclusion $G \subseteq
\text{cplt}(H)$ means that all pairs of adjacent vertices in $H$ commute.

The following definition introduces two classical finite point-line geometries
which will serve as examples of hypergrams.

\begin{definition}
The \emph{doily} is the triangle-free self-dual finite incidence geometry
composed of 15 points and 15 lines, with three points on a line and, dually,
three lines through a point. In~\Cref{figTwospread} the doily is represented by
all the lines, either dashed or plain. A \emph{two-spread} is a point-line
geometry obtained from the doily by removing a \emph{spread}, i.\,e., a set of
hyperedges covering every vertex exactly once. In~\Cref{figTwospread} the
removed spread is represented by the dashed lines, and the two-spread by the
plain lines.
\end{definition}

When $G = \text{cplt}(H)$, the hypergram $(V,H,G)$ is identified with its
hypergraph $(V,H)$, as in~\Cref{exDoily}. Under the more restrictive conditions
of magic sets, these hypergraphs $(V,H)$ are considered in~\cite{TLC22}.

\begin{example}[Doily]\label{exDoily}
 As a hypergram, the doily is $S_d =
(\{1,\ldots,15\},H_d,G_d)= (\{1,\ldots,15\},H_d)$ with $H_d = \{$
$\{1,2,3\}$, $\{1,8,9\}$, $\{1,10,11\}$, $\{2,4,6\}$, $\{2,5,7\}$,
$\{3,12,15\}$, $\{3,13,14\}$, $\{4,8,12\}$, $\{4,10,14\}$, $\{5,8,13\}$,
$\{5,10,15\}$, $\{6,9,15\}$, $\{6,11,13\}$, $\{7,9,14\}$, $\{7,11,12\}$ $\}$ 
and $G_d = \text{cplt}(H_d)$.
\end{example}

The anticommutation graph $G$ added in our definition extends the framework to a
wider range of cases, when $G \subsetneq \text{cplt}(H)$, as illustrated
by~\Cref{ex1} and detailed in~\Cref{msgdh23Sec}.

\begin{example}[Running example]\label{ex1}
As running example, let us consider the hypergram $S_{2s} =
(V_{2s},H_{2s},G_{2s})$, called \emph{two-spread hypergram}, with the set of 15
vertices $V_{2s} = \{1,\ldots,15\}$, the set of ten hyperedges $H_{2s} = \{$
$\{1,2,3\}$, $\{1,10,11\}$, $\{2,4,6\}$, $\{3,13,14\}$, $\{4,8,12\}$,
$\{5,8,13\}$, $\{5,10,15\}$, $\{6,9,15\}$, $\{7,9,14\}$, $\{7,11,12\}$ $\}$ and
the set of anticommutations $G_{2s} = \{$ $\{1,4\}$, $\{1,5\}$, $\{1,6\}$,
$\{1,7\}$, $\{1,12\}$, $\{1,13\}$, $\{1,14\}$, $\{1,15\}$, $\{2,8\}$, $\{2,9\}$,
$\{2,10\}$, $\{2,11\}$, $\{2,12\}$, $\{2,13\}$, $\{2,14\}$, $\{2,15\}$,
$\{3,4\}$, $\{3,5\}$, $\{3,6\}$, $\{3,7\}$, $\{3,8\}$, $\{3,9\}$, $\{3,10\}$,
$\{3,11\}$, $\{4,5\}$, $\{4,7\}$, $\{4,9\}$, $\{4,11\}$, $\{4,13\}$, $\{4,15\}$,
$\{5,6\}$, $\{5,9\}$, $\{5,11\}$, $\{5,12\}$, $\{5,14\}$, $\{6,7\}$, $\{6,8\}$,
$\{6,10\}$, $\{6,12\}$, $\{6,14\}$, $\{7,8\}$, $\{7,10\}$, $\{7,13\}$,
$\{7,15\}$, $\{8,10\}$, $\{8,11\}$, $\{8,14\}$, $\{8,15\}$, $\{9,10\}$,
$\{9,11\}$, $\{9,12\}$, $\{9,13\}$, $\{10,12\}$, $\{10,13\}$, $\{11,14\}$,
$\{11,15\}$, $\{12,13\}$, $\{12,14\}$, $\{13,15\}$, $\{14,15\}$ $\}$. For
instance, the edge $\{1,8\}$ is in $\text{cplt}(H_{2s})$ but not in $G_{2s}$, so
$G_{2s} \subsetneq \text{cplt}(H_{2s})$ in this case.
\end{example}

\begin{figure}[htb!]
\begin{center}
\begin{tikzpicture}[every plot/.style={smooth, tension=2},
  scale=2.5,
  every node/.style={scale=0.7,circle,draw=black,fill=white,minimum size=1.0cm}
]
  \coordinate (IX) at (0,1.0);
  \coordinate (ZX) at (0,-0.80);
  \coordinate (ZI) at (0,-0.40);
  \coordinate (XX) at (-0.95,0.30);
  \coordinate (ZZ) at (0.76,-0.24);
  \coordinate (YY) at (0.38,-0.12);
  \coordinate (YZ) at (-0.58,-0.80);
  \coordinate (YI) at (0.47,0.65);
  \coordinate (IZ) at (0.23,0.32);
  \coordinate (XY) at (0.58,-0.80);
  \coordinate (XI) at (-0.47,0.65);
  \coordinate (IY) at (-0.23,0.32);
  \coordinate (YX) at (0.95,0.30);
  \coordinate (ZY) at (-0.76,-0.24);
  \coordinate (XZ) at (-0.38,-0.12);

\draw [lightgray,dashed] (YX) -- (XZ) -- (ZY);
\draw [lightgray,dashed] (XX) -- (YY) -- (ZZ);
\draw [lightgray,dashed] (IX) -- (ZI) -- (ZX);
\draw [lightgray,dashed] (XI) -- (IY) -- (XY);
\draw [lightgray,dashed] (YZ) -- (IZ) -- (YI);

\draw (IX) -- (XI) -- (XX);
\draw [style=double] (YZ) -- (ZX) -- (XY);
\draw (XY) -- (ZZ) -- (YX);
\draw (YX) -- (YI) -- (IX);
\draw (XX) -- (ZY) -- (YZ);

\draw plot coordinates{(ZI) (ZY) (IY)};
\draw plot coordinates{(IZ) (ZZ) (ZI)};
\draw plot coordinates{(YY) (ZX) (XZ)};
\draw plot coordinates{(IY) (YI) (YY)};
\draw plot coordinates{(XZ) (XI) (IZ)};

\node at (ZY) {$\vqa{13}{ZY}$};
\node at (XX) {$\vqa{3}{XX}$};
\node at (YZ) {$\vqa{14}{YZ}$};
\node at (ZZ) {$\vqa{12}{ZZ}$};
\node at (IX) {$\vqa{1}{IX}$};
\node at (YY) {$\vqa{15}{YY}$};
\node at (XI) {$\vqa{2}{XI}$};
\node at (YX) {$\vqa{11}{YX}$};
\node at (XZ) {$\vqa{6}{XZ}$};
\node at (ZI) {$\vqa{8}{ZI}$};
\node at (IY) {$\vqa{5}{IY}$};
\node at (XY) {$\vqa{7}{XY}$};
\node at (YI) {$\vqa{10}{YI}$};
\node at (IZ) {$\vqa{4}{IZ}$};
\node at (ZX) {$\vqa{9}{ZX}$};
\end{tikzpicture}
\end{center}
\caption{Illustration of the two-spread hypergram $S_{2s}$ defined in~\Cref{ex1}.
Each circled node is labeled
by a vertex $i$ in $V_{2s}$, a semi-colon, and the Pauli observable
$\alpha_{2s}(i)$ assigned to the vertex $i$ by the 2-qubit Pauli assignment
$\alpha_{2s}$ presented in~\Cref{assignQCsec}. Each hyperedge is represented by
a single or double continuous line, either straight or curved. It is composed of
three vertices. The negative context is represented by a double line. The set
$G_{2s}$ of anticommutation edges is composed of all pairs of vertices not
belonging to a common continuous or dashed line, either simple or double.}
\label{figTwospread}
\end{figure}
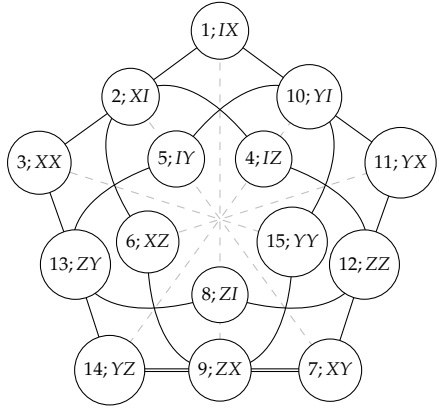

\subsection{Assignments and quantum configurations}
\label{assignQCsec}

A \emph{($n$-qubit) Pauli assignment} of a hypergram $(V,H,G)$ is an injective
function $\alpha$ from $V$ to $\{I,X,Y,Z\}^{\otimes n}-\{I^{\otimes n} \}$ that
assigns a distinct $n$-qubit Pauli observable (different from identity) to all
its vertices, such that two distinct Pauli observables $\alpha(v_1)$ and
$\alpha(v_2)$ anticommute if and only if $\{v_1, v_2\} \in G$ (\emph{commutation
condition}), and the product of all assignments of vertices in any hyperedge $h
\in H$ is the identity matrix or its opposite (formally, $\prod_{v \in h}
\alpha(v) = \pm I^{\otimes n}$) (\emph{product condition}).

\addtocounter{example}{-1}
\begin{example}[continued] An example of 2-qubit Pauli assignment of the
two-spread hypergram $S_{2s}$ is the function $\alpha_{2s} : V_{2s} \rightarrow
\{I,X,Y,Z\}^{\otimes 2} - \{I^{\otimes 2} \}$ defined by $\alpha_{2s}(1)=IX$,
$\alpha_{2s}(2)=XI$, $\alpha_{2s}(3)=XX$, $\alpha_{2s}(4)=IZ$,
$\alpha_{2s}(5)=IY$, $\alpha_{2s}(6)=XZ$, $\alpha_{2s}(7)=XY$,
$\alpha_{2s}(8)=ZI$, $\alpha_{2s}(9)=ZX$, $\alpha_{2s}(10)=YI$,
$\alpha_{2s}(11)=YX$, $\alpha_{2s}(12)=ZZ$, $\alpha_{2s}(13)=ZY$,
$\alpha_{2s}(14)=YZ$, and $\alpha_{2s}(15)=YY$. This assignment is illustrated
by a labeling in~\Cref{figTwospread}.
\end{example}

A Pauli assignment corresponds to a ``quantum satisfying assignment'' of a
binary constraint system~\cite{Cleve2014} and to a ``finite dimensional operator
solution'' to a binary linear system~\cite{CLS17}, whose matrix is the incidence
matrix of the hypergraph $H$. It also corresponds to a generalization of a
``quantum realization of a signed arrangement''~\cite{Arkhipov2012}, without the
condition that each vertex is included in exactly two hyperedges of the
hypergraph that depicts the contexts. Other works~\cite{CSW14,AFL+15} consider
graphs and hypergraphs whose vertices represent measurement outcomes, while we
consider hypergraphs whose vertices are labeled by Pauli observables. A
correspondence between these approaches can be found in~\cite[Appendix
D]{AFL+15}. This notion is also close to that of a Pauli-based
assignment~\cite{TLC22v2} (see~\Cref{seCcompare} for details) and to the
following one, coming from our previous work.

A \emph{quantum configuration}~\cite{MSGDH24} (called
\emph{quantum geometry} in~\cite{DHGMS22}) is a pair $(O,C)$ where $O$ is a
non-empty finite set of observables ($2^n$-dimensional Hermitian operators) and
$C$ is a finite set of non-empty subsets of $O$, called \emph{contexts}, such
that (i) each observable $a \in O$ satisfies $a^2=I^{\otimes n}$ (so, its
eigenvalues are in $\{-1,1\}$); (ii) any two observables $a$ and $b$ in the same
context commute, i.\,e., $a \ctimes b=b \ctimes a$; (iii) the product of all
observables in each context is either $I^{\otimes n}$ (\emph{positive} context)
or $-I^{\otimes n}$ (\emph{negative} context). In all that follows, these
observables are always phase-free Pauli observables.

\addtocounter{example}{-1}
\begin{example}[continued] \Cref{figTwospread} also shows a two-spread as a
quantum configuration, with only one negative context, represented by a double
line. The product of the 3 observables in this line is equal to $-II$, as
opposed to $+II$ for all the other lines.
\end{example}

A pair composed of a hypergram $(V,H,G)$ and an $n$-qubit Pauli assignment
$\alpha$ of it can be associated to any quantum configuration $(O,C)$ with $|O|$
$n$-qubit phase-free Pauli observables, as follows. Let $V=\{1,2,\ldots,|O|\}$.
Let $\alpha$ be a bijection from $V$ to $O$. With a small abuse of notation, let
us also denote by $\alpha$ the extension of $\alpha$ to subsets of $V$, and the
extension of the latter to subsets of subsets of $V$. Let $H$ be the inverse
image of $C$ by this last extension $\alpha$. In other words, $H$ is the set of
subsets $h$ of $V$ such that $v$ and $v'$ are in $h$ if and only if $\alpha(v)$
and $\alpha(v')$ are in the same context in $C$. Let $G$ be defined by $\{v,v'\}
\in G$ if and only if $\alpha(v)$ and $\alpha(v')$ anticommute. Then $(V,H,G)$
is a hypergram and $\alpha$ is an $n$-qubit Pauli assignment on it.
An important part of the article establishes conditions under which
the quantum notion of degree of contextuality, initially defined on a
quantum configuration (as detailed in~\Cref{contextualitySec}) becomes an
abstract notion on the corresponding hypergram, as
summarized in~\Cref{abstractCdegDef}.

Conversely, the \emph{quantum configuration associated to} a Pauli assignment
$\alpha$ of a hypergram $(V,H,G)$ is the pair $(O,C)$ such that $O = \alpha(V)$
and $C = \alpha(H)$. In all that follows, most of the notions associated with a
hypergram have their counterpart for the corresponding quantum configuration
through this correspondence.

Note that a consequence of the commutation condition and the inclusion condition
$G \subseteq \text{cplt}(H)$ for any hypergram $(V,H,G)$ is that all elements of
a context (image of a hyperedge by a Pauli assignment $\alpha$) mutually
commute. This is why this condition is not in our definition of a Pauli
assignment.

\subsection{Sign functions}
\label{signSec}

For any subset $S$ of the Pauli group whose elements mutually commute,
commutativity enables us to define $\Pi_S$ as the generalized product $\Pi_{s
\in S}\, s$ of all the elements of $S$.

Let $\alpha$ be an $n$-qubit Pauli assignment of a hypergram $(V,H,G)$, and $h$
a hyperedge in $H$. Since all elements in any context $\alpha(h)$ mutually
commute, the product $\Pi_{\alpha(h)}$ is well-defined. The \emph{sign} (or
\emph{valuation}) (\emph{function}) of $\alpha$ is the function
$\text{sgn}_{\alpha} : H \rightarrow \{-1,1\}$ such that $\Pi_{\alpha(h)} =
\text{sgn}_{\alpha}(h)\,I^{\otimes n}$ for all hyperedges $h$ in $H$.
Similarly, the sign function for a quantum configuration $(O,C)$ of $n$-qubit
observables is the function $s : C \rightarrow \{-1,1\}$ defined by $\Pi_{c} =
s(c)\,I^{\otimes n}$ for each context $c \in C$.

A \emph{classical assignment} $a : V \rightarrow \{-1,+1\}$ assigns a value $\pm
1$ to each vertex of a hypergram $(V,H,G)$. The \emph{sign (function)} of the
classical assignment $a$ is the function
$\text{sgn}_{a} : H \rightarrow \{-1,1\}$ defined by $\text{sgn}_{a}(h) =
\Pi_{a(h)} = \Pi_{v \in h}~a(v)$ for all hyperedges $h$ in $H$.

The hyperedge $h$ in $H$ is said to be \emph{satisfied}
(resp., \emph{unsatisfied}) when the signs $\text{sgn}_{\alpha}(h)$ and
$\text{sgn}_{a}(h)$ of its Pauli and classical assignments are the same
(resp., differ, and thus are opposite). These definitions also apply to contexts
and lines.

\begin{example}
In~\Cref{figDoilyAlgo:a,figDoilyAlgo:b,figDoilyAlgo:c}, the classical assignment
is represented by the numbers below the observables in each node. The dashed
lines are the unsatisfied lines on which the signs of the Pauli and classical
assignments differ.
\end{example}

In the above-mentioned approach of constraint/linear
systems~\cite{Cleve2014,CLS17}, the second member of these systems is composed
of the values of a given sign function. Similarly, by definition, a quantum
realization of a signed arrangement~\cite{Arkhipov2012} should also satisfy a
given sign function. A significant difference in~\cite{TLC22} and in the present
work is that we address a potentially simpler problem, namely finding a Pauli
assignment without targeting any given sign function. We precisely show
in~\Cref{seCproof} that satisfying a given sign function is a too strong
constraint for the objective we pursue of finding remarkable state-independent
Kochen-Specker proofs.

\subsection{Contextuality degree and noncontextual bound}
\label{contextualitySec}

With the former definitions, the contextuality of a Pauli assignment can be
defined and quantified by a natural number, such as the ``degree of
contextuality''~\cite{MSGDH24} or the ``noncontextual bound''~\cite{Cab10}, with
the following definitions and relation between them.

\begin{definition}
A Pauli assignment $\alpha$ of a hypergram $(V,H,G)$ is \emph{contextual} if
there is no classical assignment $a$ with the same sign function as $\alpha$
(over $H$). The \emph{contextuality degree}~\cite{DHGMS22} $d$ of the Pauli
assignment $\alpha$ for the hypergram $(V,H,G)$ is the minimal Hamming distance
(i.\,e., number of different values) between its sign function
$\text{sgn}_{\alpha}$ and the sign function $\text{sgn}_{a}$ of any classical
assignment $a : V \rightarrow \{-1,+1\}$.
\end{definition}

In other words, the degree of contextuality is the minimal number of different
hyperedge products between this Pauli assignment and any classical assignment.
For instance, we shall see (\Cref{2spreadProp}) that the contextuality degree of
all two-spreads is 1, meaning that at least one product will always be different
between any classical and Pauli assignments of a two-spread.

For any quantum configuration $(O,C)$, let
\begin{equation}
\chi = \sum_{c\in C} s(c) \, \langle c\rangle
\label{chiEq}
\end{equation}
be the sum of the expectation values $\langle c\rangle$ of all contexts $c$,
multiplied by their sign. A \emph{Non-Contextual Hidden-Variable theory} is a
theory in which the values of the physical observables are the same
irrespectively of the experimental context which they belong to (NCHV
hypothesis). Without this NCHV hypothesis, all the sign constraints can be
satisfied, with the expectation value $+1$ for positive contexts and $-1$ for
negative ones, so the upper bound for $\chi$ is the number $|C|$ of contexts of
$(O,C)$. However, under the NCHV hypothesis, at least $d$ sign constraints
cannot be satisfied. The expectation value of each unsatisfied context being the
opposite of its sign, the upper bound of $\chi$ under this hypothesis is reduced
by $2d$. This \emph{noncontextual bound}~\cite{Cab10,TLC22} $b$ is thus related
to $d$ by
\begin{align}
\label{bdsum}
b = |C| - 2d.
\end{align}

A quantum configuration $(O,C)$ can be transformed into an experimental
observable-based test to witness state-independent contextuality (see, e.\,g.,
\cite{KH25}). It is successful if the measurement errors are small enough to
measure a value of $\chi$ above its upper bound $(|C| - 2d)$ under the NCHV
hypothesis. Furthermore, a quantum configuration $(O,C)$ can be turned into a
non-local game, providing a proof of nonlocality of quantum physics, as detailed
in~\cite{Cleve2014,CLS17}, since it can be associated with a binary constraint
system.

\section{All Pauli assignments have the same contextuality degree}
\label{seCproof}

This section shows how to transfer a classical assignment between two Pauli
assignments of the same hypergram (\Cref{assignmentCopy}) and proves that all
these Pauli assignments have the same contextuality degree (\Cref{mainResult}),
a generalization of~\cite[Proposition 14]{TLC22v2} to all hypergrams.

Let the \emph{tensor product of two Pauli assignments} $\alpha_1$ and $\alpha_2$
of the same hypergram $(V,H,G)$ be the Pauli assignment $\alpha_{1\otimes2}$
defined by $\alpha_{1\otimes2}(v) = \alpha_1(v) \otimes \alpha_2(v)$ for all
vertices $v$ in $V$. \Cref{figDoilyAlgo} presents an example of tensor product
of two Pauli assignments of the doily structure. It will serve to illustrate the
proof arguments for~\Cref{assignmentCopy} and~\Cref{mainResult}.

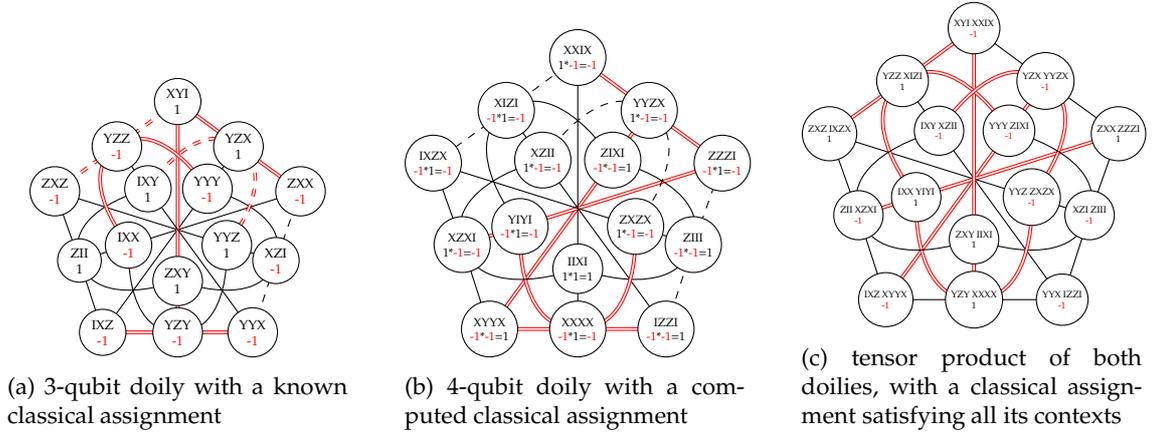
\begin{figure}[htbp!]
    \centering
    \begin{subfigure}{0.3\textwidth}
        \centering
\begin{tikzpicture}[every plot/.style={smooth, tension=2},
    scale=1.7,
    every node/.style={scale=0.5,fill=white,circle,draw=black, minimum size=1.1cm}
  ]
  \coordinate (IY) at (0,1.0);
  \coordinate (YY) at (0,-0.80);
  \coordinate (YI) at (0,-0.40);
  \coordinate (XI) at (-0.95,0.30);
  \coordinate (IZ) at (0.76,-0.24);
  \coordinate (XZ) at (0.38,-0.12);
  \coordinate (XX) at (-0.58,-0.80);
  \coordinate (ZY) at (0.47,0.65);
  \coordinate (YZ) at (0.23,0.32);
  \coordinate (ZZ) at (0.58,-0.80);
  \coordinate (XY) at (-0.47,0.65);
  \coordinate (YX) at (-0.23,0.32);
  \coordinate (ZI) at (0.95,0.30);
  \coordinate (IX) at (-0.76,-0.24);
  \coordinate (ZX) at (-0.38,-0.12);
    \draw (XI) -- (IX) -- (XX);
    \draw (ZI) -- (ZX) -- (IX);
    \draw plot coordinates{(YI) (IX) (YX)};
    \draw (XI) -- (XZ) -- (IZ);
    \draw [cRed,style=double,dashed] (IY) -- (XY) -- (XI);
    \draw [cRed,style=double] (XX) -- (YY) -- (ZZ);
    \draw (XX) -- (YZ) -- (ZY);
    \draw[dashed] (ZZ) -- (IZ) -- (ZI);
    \draw plot coordinates{(YZ) (IZ) (YI)};
    \draw [cRed,style=double] (ZI) -- (ZY) -- (IY);
    \draw [cRed,style=double] (IY) -- (YI) -- (YY);
    \draw plot coordinates{(XZ) (YY) (ZX)};
    \draw [cRed,style=double,dashed] plot coordinates{(YX) (ZY) (XZ)};
    \draw [cRed,style=double] plot coordinates{(ZX) (XY) (YZ)};
    \draw (XY) -- (YX) -- (ZZ);
    \node[align=center] at (IX) {ZII\\1};
    \node[align=center] at (XI) {ZXZ\\\textcolor{cRed}{-1}};
    \node[align=center] at (XX) {IXZ\\\textcolor{cRed}{-1}};
    \node[align=center] at (IZ) {XZI\\\textcolor{cRed}{-1}};
    \node[align=center] at (IY) {XYI\\1};
    \node[align=center] at (XZ) {YYZ\\1};
    \node[align=center] at (XY) {YZZ\\\textcolor{cRed}{-1}};
    \node[align=center] at (ZI) {ZXX\\\textcolor{cRed}{-1}};
    \node[align=center] at (ZX) {IXX\\\textcolor{cRed}{-1}};
    \node[align=center] at (YI) {ZXY\\1};
    \node[align=center] at (YX) {IXY\\1};
    \node[align=center] at (ZZ) {YYX\\\textcolor{cRed}{-1}};
    \node[align=center] at (ZY) {YZX\\1};
    \node[align=center] at (YZ) {YYY\\\textcolor{cRed}{-1}};
    \node[align=center] at (YY) {YZY\\\textcolor{cRed}{-1}};
  \end{tikzpicture}
        \caption{3-qubit doily with a known classical assignment}
        \label{figDoilyAlgo:a}
    \end{subfigure}
    \hfill
    \begin{subfigure}{0.3\textwidth}
        \centering
\begin{tikzpicture}[every plot/.style={smooth, tension=2},
    scale=2,
    every node/.style={scale=0.5,fill=white,circle,draw=black, minimum size=1.3cm}
  ]
  \coordinate (xxix) at (0,1.0);
  \coordinate (xxxx) at (0,-0.80);
  \coordinate (iixi) at (0,-0.40);
  \coordinate (ixzx) at (-0.95,0.30);
  \coordinate (ziii) at (0.76,-0.24);
  \coordinate (zxzx) at (0.38,-0.12);
  \coordinate (xyyx) at (-0.58,-0.80);
  \coordinate (yyzx) at (0.47,0.65);
  \coordinate (zixi) at (0.23,0.32);
  \coordinate (izzi) at (0.58,-0.80);
  \coordinate (xizi) at (-0.47,0.65);
  \coordinate (xzii) at (-0.23,0.32);
  \coordinate (zzzi) at (0.95,0.30);
  \coordinate (xzxi) at (-0.76,-0.24);
  \coordinate (yiyi) at (-0.38,-0.12);
  \draw[dashed]  (xxix) -- (xizi) -- (ixzx);
  \draw  (ixzx) -- (xzxi) -- (xyyx);
  \draw[double,draw=cRed]  (xyyx) -- (xxxx) -- (izzi);
  \draw[dashed]  (izzi) -- (ziii) -- (zzzi);
  \draw[double,draw=cRed]  (zzzi) -- (yyzx) -- (xxix);
  \draw[double,draw=cRed] (xyyx) -- (zixi) -- (yyzx);
  \draw  (xizi) -- (xzii) -- (izzi);
  \draw[double,draw=cRed]  (zzzi) -- (yiyi) -- (xzxi);
  \draw  (ixzx) -- (zxzx) -- (ziii);
  \draw  (xxix) -- (iixi) -- (xxxx);
  \draw  plot coordinates{(yiyi) (xizi) (zixi)};
  \draw[dashed]  plot coordinates{(xzii) (yyzx) (zxzx)};
  \draw  plot coordinates{(zixi) (ziii) (iixi)};
  \draw[double,draw=cRed]  plot coordinates{(zxzx) (xxxx) (yiyi)};
  \draw  plot coordinates{(iixi) (xzxi) (xzii)};
  \small
  \node[align=center]  at (xxix) {XXIX\\1*\textcolor{cRed}{-1}=\textcolor{cRed}{-1}};
  \node[align=center]  at (xxxx) {XXXX\\\textcolor{cRed}{-1}*1=\textcolor{cRed}{-1}};
  \node[align=center]  at (iixi) {IIXI\\1*1=1};
  \node[align=center]  at (ixzx) {IXZX\\\textcolor{cRed}{-1}*1=\textcolor{cRed}{-1}};
  \node[align=center]  at (ziii) {ZIII\\\textcolor{cRed}{-1}*\textcolor{cRed}{-1}=1};
  \node[align=center]  at (zxzx) {ZXZX\\1*\textcolor{cRed}{-1}=\textcolor{cRed}{-1}};
  \node[align=center]  at (xyyx) {XYYX\\\textcolor{cRed}{-1}*\textcolor{cRed}{-1}=1};
  \node[align=center]  at (yyzx) {YYZX\\1*\textcolor{cRed}{-1}=\textcolor{cRed}{-1}};
  \node[align=center]  at (zixi) {ZIXI\\\textcolor{cRed}{-1}*\textcolor{cRed}{-1}=1};
  \node[align=center]  at (izzi) {IZZI\\\textcolor{cRed}{-1}*\textcolor{cRed}{-1}=1};
  \node[align=center]  at (xizi) {XIZI\\\textcolor{cRed}{-1}*1=\textcolor{cRed}{-1}};
  \node[align=center]  at (xzii) {XZII\\1*\textcolor{cRed}{-1}=\textcolor{cRed}{-1}};
  \node[align=center]  at (zzzi) {ZZZI\\\textcolor{cRed}{-1}*1=\textcolor{cRed}{-1}};
  \node[align=center]  at (xzxi) {XZXI\\1*\textcolor{cRed}{-1}=\textcolor{cRed}{-1}};
  \node[align=center]  at (yiyi) {YIYI\\\textcolor{cRed}{-1}*1=\textcolor{cRed}{-1}};
  \end{tikzpicture}
        \caption{4-qubit doily with a computed classical assignment}
        \label{figDoilyAlgo:b}
    \end{subfigure}
    \hfill
    \begin{subfigure}{0.3\textwidth}
        \centering
\begin{tikzpicture}[every plot/.style={smooth, tension=2},
    scale=2,
    every node/.style={scale=0.5,fill=white,circle,draw=black, minimum size=1.3cm}
  ]
  \coordinate (xxix) at (0,1.0);
  \coordinate (xxxx) at (0,-0.80);
  \coordinate (iixi) at (0,-0.40);
  \coordinate (ixzx) at (-0.95,0.30);
  \coordinate (ziii) at (0.76,-0.24);
  \coordinate (zxzx) at (0.38,-0.12);
  \coordinate (xyyx) at (-0.58,-0.80);
  \coordinate (yyzx) at (0.47,0.65);
  \coordinate (zixi) at (0.23,0.32);
  \coordinate (izzi) at (0.58,-0.80);
  \coordinate (xizi) at (-0.47,0.65);
  \coordinate (xzii) at (-0.23,0.32);
  \coordinate (zzzi) at (0.95,0.30);
  \coordinate (xzxi) at (-0.76,-0.24);
  \coordinate (yiyi) at (-0.38,-0.12);
  \draw[double,draw=cRed]  (xxix) -- (xizi) -- (ixzx);
  \draw  (ixzx) -- (xzxi) -- (xyyx);
  \draw  (xyyx) -- (xxxx) -- (izzi);
  \draw  (izzi) -- (ziii) -- (zzzi);
  \draw  (zzzi) -- (yyzx) -- (xxix);
  \draw[double,draw=cRed]  (xyyx) -- (zixi) -- (yyzx);
  \draw  (xizi) -- (xzii) -- (izzi);
  \draw[double,draw=cRed]  (zzzi) -- (yiyi) -- (xzxi);
  \draw  (ixzx) -- (zxzx) -- (ziii);
  \draw[double,draw=cRed]  (xxix) -- (iixi) -- (xxxx);
  \draw[double,draw=cRed]  plot coordinates{(yiyi) (xizi) (zixi)};
  \draw[double,draw=cRed]  plot coordinates{(xzii) (yyzx) (zxzx)};
  \draw  plot coordinates{(zixi) (ziii) (iixi)};
  \draw[double,draw=cRed]  plot coordinates{(zxzx) (xxxx) (yiyi)};
  \draw  plot coordinates{(iixi) (xzxi) (xzii)};
  \scriptsize
  \node[align=center]  at (xxix) {XYI XXIX\\\color{cRed}-1};
  \node[align=center]  at (xxxx) {YZY XXXX\\1};
  \node[align=center]  at (iixi) {ZXY IIXI\\1};
  \node[align=center]  at (ixzx) {ZXZ IXZX\\1};
  \node[align=center]  at (ziii) {XZI ZIII\\\color{cRed}-1};
  \node[align=center]  at (zxzx) {YYZ ZXZX\\\color{cRed}-1};
  \node[align=center]  at (xyyx) {IXZ XYYX\\\color{cRed}-1};
  \node[align=center]  at (yyzx) {YZX YYZX\\\color{cRed}-1};
  \node[align=center]  at (zixi) {YYY ZIXI\\\color{cRed}-1};
  \node[align=center]  at (izzi) {YYX IZZI\\\color{cRed}-1};
  \node[align=center]  at (xizi) {YZZ XIZI\\1};
  \node[align=center]  at (xzii) {IXY XZII\\\color{cRed}-1};
  \node[align=center]  at (zzzi) {ZXX ZZZI\\1};
  \node[align=center]  at (xzxi) {ZII XZXI\\\color{cRed}-1};
  \node[align=center]  at (yiyi) {IXX YIYI\\1};
  \end{tikzpicture}
        \caption{tensor product of both doilies, with a classical assignment satisfying all its contexts}
        \label{figDoilyAlgo:c}
    \end{subfigure}
    \caption{Illustration of the classical assignment transfer process, between two
    Pauli assignments of the same structure, here the doily. The dashed lines
    are the unsatisfied ones.\label{figDoilyAlgo}}
\end{figure}

\begin{lemma}
\label{assignmentCopy}
Let $\alpha_1$ and $\alpha_2$ be two Pauli assignments of the same hypergram
$(V,H,G)$. Let $a_1$ be a classical assignment of $V$ and $S$ the subset of
hyperedges in $H$ satisfied by $a_1$ for $\alpha_1$. Then there exists a
classical assignment $a_2$ whose set of hyperedges satisfied for $\alpha_2$ is
also $S$.
\end{lemma}

\begin{proof}
\noindent For any two vertices $v_1$ and $v_2$ in $V$, either the two pairs
$(\alpha_1(v_1),\alpha_1(v_2))$ and $(\alpha_2(v_1),\alpha_2(v_2))$ of their
images by $\alpha_1$ and $\alpha_2$ anticommute (if $\{v_1,v_2\} \in G$), or
they both commute. By elementary algebraic computations, with $s = \pm 1$,
\begin{align*}
\alpha_{1\otimes2}(v_1) \ctimes \alpha_{1\otimes2}(v_2)
  & = (\alpha_1(v_1) \otimes \alpha_2(v_1)) \ctimes (\alpha_1(v_2) \otimes \alpha_2(v_2))
     = (\alpha_1(v_1) \ctimes \alpha_1(v_2))\otimes (\alpha_2(v_1) \ctimes \alpha_2(v_2))
\\
  & = (s\; \alpha_1(v_2) \ctimes \alpha_1(v_1)) \otimes (s\; \alpha_2(v_2) \ctimes \alpha_2(v_1))
\\
  & = s^2\; (\alpha_1(v_2) \ctimes \alpha_1(v_1))\otimes (\alpha_2(v_2) \ctimes \alpha_2(v_1))
\\
  & = (\pm 1)^2\; (\alpha_1(v_2) \otimes \alpha_2(v_2)) \ctimes (\alpha_1(v_1) \otimes \alpha_2(v_1))
\\
  & =  \alpha_{1\otimes2}(v_2) \ctimes \alpha_{1\otimes2}(v_1).
\end{align*}
\noindent  This means that all
observables in the image $O \equiv \alpha_{1\otimes2}(V)$ of
$\alpha_{1\otimes2}$ pairwise commute.

With $C = \alpha_{1\otimes2}(H)$, the quantum configuration $(O,C)$ is therefore
commutative. By~\cite[Proposition 8]{Arkhipov2012}, it is non-contextual. Let
$a_{1\otimes2}$ be a classical assignment satisfying all sign constraints of
$(O,C)$, and let $a_2$ be the classical assignment of $V$ defined by $a_2(v) =
a_{1\otimes2}(v) a_1(v)$.

Then, by the fact that for all hyperedges $h \in H$,
$\prod_{v\in h}a_2(v) = \prod_{v\in h}(a_{1\otimes2}(v) a_1(v))$, we obtain that
the classical assignment $a_2$ satisfies for $\alpha_2$ exactly the same
hyperedges as $a_1$ for $\alpha_1$.
\end{proof}

\Cref{figDoilyAlgo} illustrates the operational aspect of~\Cref{assignmentCopy}
and its proof, as a way to transfer a classical assignment from a given Pauli
assignment to another one with the same structure. Assume we already know a
classical assignment $a_1$ reaching the contextuality degree in the 3-qubit
doily in~\Cref{figDoilyAlgo:a}. \Cref{figDoilyAlgo:c} shows the tensor product
of the first two doilies, for which a non-contextual solution $a_{1\otimes2}$ is
easily computed. Finally, the product of the classical assignment $a_1$
of~\Cref{figDoilyAlgo:a} and the solution $a_{1\otimes2}$
of~\Cref{figDoilyAlgo:c} provides an optimal classical assignment $a_2$ for the
4-qubit doily in~\Cref{figDoilyAlgo:b}, with the same subset of satisfied
hyperedges, so the same contextuality degree, as generalized in the following
theorem.

\begin{theorem}
\label{mainResult}
Let $(V,H,G)$ be a hypergram. Then all Pauli assignments of $(V,H,G)$ have
the same contextuality degree and noncontextual bound.
\end{theorem}

\begin{proof}
When $(V,H,G)$ admits no Pauli assignment, the theorem trivially holds.
Otherwise, let $\alpha_1$ be any Pauli assignment of $(V,H,G)$, with the
contextuality degree $d_1$, and let $a_1$ be a classical assignment of
$\alpha_1$ for this contextuality degree $d_1$, i.\,e., at the minimal Hamming
distance $d_1$ from $\alpha_1$. Let $\alpha_2$ be another Pauli assignment whose
contextuality degree $d_2$ is unknown. By~\Cref{assignmentCopy}, we know that
there is a classical assignment $a_2$ with the same set of unsatisfied
hyperedges for $\alpha_2$ as $a_1$ for $\alpha_1$, and thus at the same Hamming
distance $d_1$ from $\alpha_2$, which means that the contextuality degree $d_2$
of $\alpha_2$ is at most $d_1$. The same reasoning with $\alpha_1$ and
$\alpha_2$ exchanged entails that $d_1$ is at most $d_2$, so $d_1 = d_2$.

The noncontextual bound $b$ being related to $d$ by the linear
relation~\eqref{bdsum}, $\alpha_1$ and $\alpha_2$ also have the same
noncontextual bound.
\end{proof}

\section{Assignability}
\label{seCcns}

By \Cref{mainResult} all Pauli assignments $\alpha$ of a hypergram have the same
contextuality degree. However, a given hypergram $(V,H,G)$ does not necessarily
admit a Pauli assignment. If it does, it is said to be
\emph{(Pauli-)assignable}. After providing a counterexample and introducing some
definitions and notations, we establish in~\Cref{vecCnsQAprop,graphCnsQA} a
necessary and sufficient condition on $H$ and $G$ for the assignability of
$(V,H,G)$. When this \emph{assignability condition} is satisfied, the algorithm
presented in~\Cref{labelingSec} efficiently computes such a Pauli assignment
$\alpha$, used in~\Cref{justifSec} to complete the proof of~\Cref{vecCnsQAprop}.
The algorithmic complexity of this algorithm is discussed
in~\Cref{algoComplexitySec}. Finally, \Cref{abstractCdegDef} proposes a
definition of contextuality degree for any assignable hypergram.

From here, we superimpose to the (hyper)graph point of view of the previous
sections the algebraic point of view promoted by algebraic graph
theory~\cite{GR01}, through the following definitions and identifications. For
this purpose, the hyperedges in the set $H$ of a context hypergraph $(V,H)$ are
assumed to be numbered from $1$ to $|H|$ in a arbitrary but fixed order.

\begin{definition}
The \emph{context matrix} $C(H) \in \mathbb{F}_2^{|H|\times|V|}$ of the context
hypergraph $(V,H)$ is its incidence matrix, defined by $C(H)_{k,v} = 1$ if the
vertex $v \in V =\{1,\ldots,|V|\}$ is in the $k$-th hyperedge ($1 \leq k \leq
|H|$), and $0$ otherwise. The \emph{anticommutation matrix} $A(G) \in
\mathbb{F}_2^{|V|\times|V|}$ of the anticommutation graph $(V,G)$ is its
adjacency matrix, defined as the symmetric matrix such that $A(G)_{i,j} = 1$ if
$\{i,j\} \in G$, and $0$ otherwise.
\end{definition}

Since the anticommutation graph $(V,G)$ is loopless, all the diagonal entries of
$A(G)$ are equal to zero.

To lighten the notations (for instance in~\Cref{hypergramSex}
and~\Cref{vecCnsQAprop}) and without risk of confusion, we adopt the following
identification conventions: we also designate by $H$ the context matrix $C(H)$
and by $G$ the anticommutation matrix $A(G)$, whenever it is clear from the
context whether we are talking about a matrix or a set.

\begin{example}\label{hypergramSex}
Consider the hypergram $S = (V,H,G)$ with $V = \{v_1,v_2,v_3,v_4,v_5\}$, $H =
\{\{v_1,v_2,v_3\}$, $\{v_1,v_4,v_5\}\}$ and $G = \{\{v_3,v_5\}\}$.
Its context matrix is the incidence matrix
\begin{equation*}
H = \begin{pmatrix}
1 & 1 & 1 & 0 & 0 \\
1 & 0 & 0 & 1 & 1
\end{pmatrix}
\end{equation*}
and its anticommutation matrix is the adjacency matrix
\begin{equation*}
G = \begin{pmatrix}
0 & 0 & 0 & 0 & 0 \\
0 & 0 & 0 & 0 & 0 \\
0 & 0 & 0 & 0 & \textbf{1} \\
0 & 0 & 0 & 0 & 0 \\
0 & 0 & \textbf{1} & 0 & 0
\end{pmatrix}.
\end{equation*}
The two $1$s in \textbf{bold} come from the anticommutation between $v_3$ and
$v_5$ in the set $G$. 
\end{example}

The following lemma from~\cite{TLC22v2} is a key argument in the
proof of~\Cref{vecCnsQAprop}. It is reproduced here (with a slight adaptation to
our notations) in order to make the paper self-contained. The \emph{Gram matrix}
of the sequence of vectors $s_1, \dots, s_m \in \mathbb{F}_2^{2n}$ with respect
to the symplectic form $\innerproduct{.}{.}$ is the $m \times m$ matrix whose
$(i,j)$ entry $G_{i,j}$ is $\innerproduct{s_i}{s_j}$.

\begin{lemma}[{\cite[Lemma 3]{TLC22v2}}]
\label{lemma3tlc} 
Let $s_1, \dots, s_m \in \mathbb{F}_2^{2n}$ for some positive integer $n$, and
let $G$ be their  Gram matrix with respect to the symplectic form
$\innerproduct{.}{.}$, with rows $r_1, \dots, r_m$. If a subset $\{r_{i_1},
\dots, r_{i_j}\}$ of the rows is linearly independent, then the corresponding
set of vectors $\{s_{i_1}, \dots, s_{i_j}\}$ is also linearly independent.
\end{lemma}

\begin{theorem}
\label{vecCnsQAprop}
A hypergram $(V,H,G)$ admits a Pauli assignment if and only if
\begin{align}
\label{matrixCNS}
H \ctimes G = 0.
\end{align} 
\end{theorem}

In~\Cref{matrixCNS}, called the \emph{(matricial) assignability condition}, the
dot is the matrix product and its right-hand side $0$ is the $|H|\times|V|$ zero
matrix. In other words, the condition is that each column of the matrix $G$ is
in the null space $\text{ker}(H)$ of the matrix $H$.

\addtocounter{example}{-1}
\begin{example}[continued] The fifth column of $G$ in our example is the vector
$(0~0~1~0~0)^T$. Since $H \ctimes (0~0~1~0~0)^T = (1~0)^T \neq (0~0)^T$, the
hypergram $S$ is not assignable.
\end{example}

\begin{proof}
First, we show that each Pauli assignment for $(V,H,G)$ satisfies the
assignability condition~(\ref{matrixCNS}).

Let $\alpha$ be an $n$-qubit Pauli assignment of $(V,H,G)$, with $V$ assumed to
be $\{1, 2, \ldots, |V|\}$ here. The commutation condition $G_{i,j} =
\innerproduct{\psi(\alpha(i))}{\psi(\alpha(j))}$ for $\alpha$ is equivalent to
the fact that $G$ is the Gram matrix of $\psi \circ \alpha$, by
definition of a Gram matrix. On the other hand, the commutation and product
conditions for $\alpha$ entail that $\psi \circ \alpha : V \rightarrow
\mathbb{F}_2^{2n}$ and $G$ satisfy the hypotheses of \Cref{lemma3tlc}, which
provides the consequence that $G$ is a valid Gram matrix for the hypergraph
$(V,H)$, as defined in~\cite[page 8]{TLC22v2}. In particular, the following
second condition for a valid Gram matrix holds, for all hyperedges $h$ in $H$
and all vertices $1 \leq v \leq |V|$:
\begin{align}
\sum_{i~\in~h} G_{i,v} = 0. \label{anticommMatrixColumnHyperedgeZero}
\end{align}
In matricial form this system of linear equations is~\Cref{matrixCNS}.

Conversely, the fact that~\Cref{matrixCNS} implies the existence of a Pauli
assignment is justified after presenting the algorithm in~\Cref{labelingSec},
which constructs such a Pauli assignment $\alpha$ of $(V,H,G)$
from the anticommutation matrix $G$.
\end{proof}

It is sometimes more convenient and more pictorial to express the assignability
condition by using the language of graphs, as in the following theorem.

\begin{theorem}
\label{graphCnsQA}
A hypergram $(V,H,G)$ admits a Pauli assignment if and only if, for each vertex
$v \in V$, each hyperedge $h \in H$ contains an even number of vertices which
are adjacent to $v$ in the graph $(V,G)$.
\end{theorem}

In other words, this \emph{graphical assignability condition} requires that each
vertex anticommutes with an even number of vertices of each hyperedge. Since all
the vertices in a hyperedge mutually commute, it is sufficient to consider the
vertices not belonging to the hyperedge.

\begin{proof}
\Cref{matrixCNS} is an algebraic form of the graphical assignability condition.
Indeed, consider the entry $(H \ctimes G)_{h,v}$ at row $h$ and column $v$ in
the product of the matrices $H$ and $G$. Each $1$ in this sum is the product of
a $1$ entry in the matrix $H$, meaning that some vertex $i$ is in the hyperedge
$h$, and a $1$ entry in the matrix $G$, meaning that this vertex $i$
anticommutes with $v$. The sum of these $1$s as natural numbers would be the
number of such vertices $i$ in $h$ anticommuting with $v$. The sum being
computed over $\mathbb{F}_2$, it is zero as required
in~\Cref{anticommMatrixColumnHyperedgeZero} if and only if this number is even.
\end{proof}

\addtocounter{example}{-1}
\begin{example}[continued] In our example where $G = \{\{v_3,v_5\}\}$ the vertex
$v_5$ anticommutes with $v_3$ and commutes with $v_1$ and $v_2$, so it
anticommutes only with one vertex in the first hyperedge $\{v_1,v_2,v_3\}$.
By~\Cref{graphCnsQA}, the hypergram $S$ is not assignable.
\end{example}

\begin{corollary}
\label{nbVerticesNullityHineq}
For all assignable hypergrams $(V,H,G)$,
\begin{align}
\label{ineqRkKer}
|V| \leq 2^{\text{rk}(\text{ker}(H))}-1.
\end{align}
\end{corollary}

\begin{proof}
Direct consequence of the assignability condition and the fact that $(V,G)$ is a
reduced graph: $G$ has no line of zeros and no duplicated columns. So, each
column of $G$ is a distinct vector in the null space of $H$. The lhs of the
inequality~(\ref{ineqRkKer}) is the number of columns of $G$. Its rhs is the
number of non-null vectors in the kernel space of $H$.
\end{proof}

\subsection{Pauli-labeling Algorithm}
\label{labelingSec}

Let $(V,H,G)$ be a hypergram whose set of vertices $V =\{1,\ldots,|V|\}$ is
totally ordered by $<$. From $V$ and the anticommutation matrix $G$,
\Cref{alg:findAssignment} modifies a copy $B$ of the input matrix $G$ until
reaching the null matrix, as in the algorithm left implicit in the proof of
Lemma 8.9.3 in~\cite{GR01}. Moreover, \Cref{alg:findAssignment} computes and
returns a function $\alpha : V \rightarrow \{I,X,Y,Z\}^{\otimes n}$ that labels
all the vertices in $V$ with $n$-qubit observables. It also returns the number
$n$ of qubits in these labels.

\begin{algorithm}[H]
  \caption{Pauli Assignment from an Anticommutation Matrix $G$ on the Set of Vertices $V$}
  \begin{algorithmic}[1]
  \Function{PauliAssignmentFromAnticommutations}{$V,G$}
      \State $n \gets 0$ \label{algFindAssignmentInit1}
      \State $B \gets G$
      \While{$i,j \gets \Call{FindOverdiagonalOne}{B}$} \label{algFindAssignmentWhile}
          \State $n \gets n+1$ \label{nUpdate}
          \For{$k \in V$} \label{beginFor}
            \State $\alpha(k)_{n} \gets \psi^{-1}\left(B_{k,i},B_{k,j}\right)$ \label{alphaNupdateB}
          \EndFor
          \State $B \gets B + B \ctimes e_i \ctimes (B \ctimes e_j)^T + B \ctimes e_j \ctimes (B \ctimes e_i)^T$ \label{matrixUpdateB1}
      \EndWhile
      \State \Return $\alpha, n$
  \EndFunction
  \end{algorithmic}
  \label{alg:findAssignment}
\end{algorithm}

Calling the function \textsc{FindOverdiagonalOne} on
Line~\ref{algFindAssignmentWhile} either returns vertices $i$ and $j$ such that
$i < j$ and $B_{i,j} = 1$, or \textsc{false} when no such pair of vertices
exists, which ends the loop. For any vertex $k \in V$ and $1 \leq s \leq t \leq
n$, let $\alpha(k)_{s}$ denote the $s$-th qubit of $\alpha(k)$ and
$\alpha(k)_{s..t}$ denote its sequence of qubits from the $s$-th one to the
$t$-th one included. For $n \geq 1$ the assignment on Line~\ref{alphaNupdateB}
computes at the $n$-th iteration of~\Cref{alg:findAssignment} the $n$-th Pauli
matrix $\alpha(k)_{n}$ of the label $\alpha(k)$ of all vertices $k \in V$, by
using the inverse $\psi^{-1}$ of the encoding function $\psi$ defined by
(\ref{paulipts}).

For $m \in V$, let $e_m$ denote the $m$-th standard basis vector. Then $B
\ctimes e_m$ is the $m$-th column of $B$ and $e_l^T \ctimes B \ctimes e_m =
B_{l,m}$ is its entry at row $l$ in this column. On Line~\ref{matrixUpdateB1}
the $i$-th and $j$-th columns of the matrix $B$ are used to add zeros in $B$, as
detailed in~\Cref{commutationCondSec}.

\subsection{Justification}
\label{justifSec}

When~\Cref{matrixCNS} holds, the present section shows that the labeling
$\alpha$ returned by \Call{PauliAssignmentFromAnticommutations}{$V,G$} is a
Pauli assignment of the hypergram $(V,H,G)$, thus completing the proof
of~\Cref{vecCnsQAprop}. More precisely, \Cref{commutationCondSec} justifies the
commutation condition, \Cref{rangeSec} justifies that the labels according to
$\alpha$ are pairwise distinct and different from the identity and
\Cref{prodCondSec} justifies the product condition under a condition on the rank
of $G$ proved in~\Cref{rankSec}.

In all that follows, when there is no risk of confusion, the function $\psi$ is
often omitted, for instance when writing $\innerproduct{\alpha(l)}{\alpha(m)}$
instead of $\innerproduct{\psi(\alpha(l))}{\psi(\alpha(m))}$ in~\Cref{inv}.

\subsubsection{Commutation condition}
\label{commutationCondSec}

The following lemma provides a key argument for the commutation condition.
A loop invariant is a property which is true before the loop and is preserved by
each iteration of the loop. Consequently, it also holds after the loop.

\begin{lemma}
\label{invariantProp}
The formula
\begin{align}
\label{inv}
\forall l,m \in V.\ B_{l,m} = G_{l,m}+\innerproduct{\alpha(l)}{\alpha(m)}
\end{align}
is an invariant for the loop of \Cref{alg:findAssignment}.
\end{lemma}

\begin{proof}
For $n \geq 1$, let $B^{(n-1)}$ denote the value of the matrix $B$ at the
beginning of the $n$-th iteration of \Cref{alg:findAssignment}, just before
Line~\ref{nUpdate}. Consequently, $B^{(n)}$ denotes the value of $B$ at the end
of the $n$-th iteration, just after Line~\ref{matrixUpdateB1}. In particular,
$B^{(0)} = G$.

The notation $\alpha(l)$ in~(\ref{inv}) stands for $\alpha(l)_{1..n}$ at
Line~\ref{matrixUpdateB1}. Its value when $n = 0$ can be chosen constant, so
that $\innerproduct{\alpha(l)}{\alpha(m)} = 0$ before the loop. Thus (\ref{inv})
is initially true, because $B = G$ before the loop. It remains to prove that
(\ref{inv}) is preserved by the assignment on Line~\ref{matrixUpdateB1}, under
the assumption
\begin{align}
\alpha(k)_{n} = \psi^{-1}\left(B^{(n-1)}_{k,i},B^{(n-1)}_{k,j}\right)
\end{align}
coming from Line~\ref{alphaNupdateB}, equivalent to 
\begin{align}
\label{psiBeq}
\psi(\alpha(k))_{2n-1} = B^{(n-1)}_{k,i} \qquad \text{ and } \qquad \psi(\alpha(k))_{2n} = B^{(n-1)}_{k,j}
\end{align}
for all vertices $k$. This preservation is justified by the following sequence
of equalities:
\begin{align*}
B_{l,m}^{(n)}
& = B_{l,m}^{(n-1)} + \left(B^{(n-1)} \ctimes e_i \ctimes (B^{(n-1)} \ctimes e_j)^T + B^{(n-1)} \ctimes e_j \ctimes (B^{(n-1)} \ctimes e_i)^T\right)_{l,m}
\\
& = B_{l,m}^{(n-1)} + \left(B^{(n-1)} \ctimes e_i \ctimes (B^{(n-1)} \ctimes e_j)^T\right)_{l,m} + \left(B^{(n-1)} \ctimes e_j \ctimes (B^{(n-1)} \ctimes e_i)^T\right)_{l,m}
\\
& = B_{l,m}^{(n-1)} + B^{(n-1)}_{l,i} B^{(n-1)}_{m,j} + B^{(n-1)}_{l,j} B^{(n-1)}_{m,i}
\\
& = B_{l,m}^{(n-1)} + \psi(\alpha(l))_{2n-1}\psi(\alpha(m))_{2n} + \psi(\alpha(l))_{2n}\psi(\alpha(m))_{2n-1} & \text{by}~(\ref{psiBeq})
\\
& = B_{l,m}^{(n-1)} + \innerproduct{\psi(\alpha(l)_n)}{\psi(\alpha(m)_n)} & \text{by}~(\ref{symplf})
\\
& = B_{l,m}^{(0)} + \innerproduct{\psi(\alpha(l)_1)}{\psi(\alpha(m)_1)} + \ldots + \innerproduct{\psi(\alpha(l)_n)}{\psi(\alpha(m)_n)}
  & \text{by induction on } n
\\
& = G_{l,m} + \innerproduct{\psi(\alpha(l)_{1..n})}{\psi(\alpha(m)_{1..n})},
\end{align*}
so (\ref{inv}) is an invariant of the loop.
\end{proof}

It is easy to check similarly that the symmetry of $B$ and its diagonal of zeros
are two other loop invariants. All these invariants still hold after the loop,
together with the negation of the loop condition. So, at the end of the
algorithm, $B$ is the null matrix and thus $G_{l,m} =
\innerproduct{\alpha(l)}{\alpha(m)}$ for all vertices $l$ and $m$, meaning that
the returned labeling $\alpha : V \rightarrow \{I,X,Y,Z\}^{\otimes n}$ satisfies
the commutation condition.

\subsubsection{Rank and number of qubits}
\label{rankSec}

The following lemma establishes the relation $r=2n$ between the rank $r$ of $G$
and the number $n$ of qubits of the labeling generated by the algorithm. It
could be justified in one sentence saying that it is the first conclusion of
Theorem 8.10.1 in~\cite{GR01}, but, in order to make the paper self-contained,
we prefer to provide a proof of it, with our notations and more details. This
proof is strongly inspired by that of Lemma 8.9.3 in the same
reference~\cite{GR01}.

\begin{lemma}
\label{rankProp}
The rank of the matrix $G$ is twice the number $n$ of iterations returned by the
application \Call{PauliAssignment\-FromAnticommutations}{$V,G$}
of~\Cref{alg:findAssignment} to $V$ and $G$.
\end{lemma}

\begin{proof}
It is sufficient to show that each execution of Line~\ref{matrixUpdateB1}
reduces the rank $\text{rk}(B)$ of $B$ by two, since the rank of the final null
matrix $B$ is zero.

Let $A$ be $B$ just before Line~\ref{matrixUpdateB1}, $y = A \ctimes e_i$ and $z
= A \ctimes e_j$ be the $i$-th and $j$-th columns of $A$ and $C = y \ctimes z^T
+ z \ctimes y^T$ be the matrix subtracted to $A$ to update $B$ on
Line~\ref{matrixUpdateB1}. Since $A_{i,j} = 1$ and $A_{i,i} = A_{j,j} = 0$, the
$i$-th column of $C$ is the $i$-th column of $A$, the $j$-th column of $C$ is
the $j$-th column of $A$ and the other columns of $C$ are linear combinations of
these two columns of $A$.

Since $A = B+C$ after Line~\ref{matrixUpdateB1}, the $m$-th column of $A$ is a
linear combination of the $m$-th column of $B$ and the $i$-th and $j$-th columns
$y$ and $z$ of $A$. So, the column space of $A$ is spanned by the union of the
columns of $B$ with the vectors $y$ and $z$, so $\text{rk}(A) \leq
\text{rk}(B)+2$.

For any vector $x$ in the null space of $A$, we have $A \ctimes x = 0$. So,
considering the $i$-th and $j$-th rows of $A$, we have $e_i^T \ctimes A \ctimes
x = 0$ and $e_j^T \ctimes A \ctimes x = 0$. Since $A$ is symmetric ($A = A^T$),
we get $(A \ctimes e_i)^T \ctimes x = 0$ and $(A \ctimes e_j)^T \ctimes x = 0$,
i.\,e., $y^T \ctimes x = z^T \ctimes x = 0$. Consequently, $C \ctimes x = (y
\ctimes z^T + z \ctimes y^T) \ctimes x = y \ctimes (z^T \ctimes x) + z \ctimes
(y^T \ctimes x) = 0$. Since $B$ just after Line~\ref{matrixUpdateB1} is $A+C$,
it comes that $B \ctimes x = 0$ and so the null space of $A$ is included in the
null space of $B$. Moreover, since $B = A-C$, the $i$-th and $j$-th columns $B
\ctimes e_i$ and $B \ctimes e_j$ of $B$ are two columns of zeros. Therefore, the
two independent basis vectors $e_i$ and $e_j$ are in the null space of $B$.
Consequently, $\text{rk}(B) \leq \text{rk}(A)-2$.

Altogether, $\text{rk}(B) = \text{rk}(A)-2$, which ends the proof.
\end{proof}

\subsubsection{Image conditions}
\label{rangeSec}

The anticommutation matrix $G$ is the Gram matrix of the function $\psi \circ
\alpha$. If $\alpha$ is not injective, then there are two distinct vertices $v$
and $v'$ such that $\psi(\alpha(v)) = \psi(\alpha(v'))$ and the corresponding
rows in $G$ are equal. Similarly, if $\alpha(v) = I^{\otimes n}$ for some
vertex $v$, then $\psi(\alpha(v)) = 0 \in \mathbb{F}_2^{2n}$ and the
corresponding row in $G$ is a row of zeros. Since the anticommutation graph
$(V,G)$ is reduced (by definition of a hypergram), its adjacency matrix $G$ has no
duplicated rows and no row of zeros, so $\alpha$ is injective and its image is
included in $\{I,X,Y,Z\}^{\otimes n}~-~\{I^{\otimes n} \}$.

\subsubsection{Product condition}
\label{prodCondSec}

Finally, the product condition is also respected, since $\alpha$ satisfies all
the hypotheses of the following lemma.

\begin{lemma}
\label{productAlphaProp}
Let $(V,H,G)$ be an assignable hypergram with $V =\{1,\ldots,|V|\}$. Let $r$ be
the rank of $G$ and $n = r/2$. Let $\alpha : V \rightarrow \{I,X,Y,Z\}^{\otimes
n}$ be a vertex labeling with $n$-qubit observables satisfying the commutation
condition $\innerproduct{\alpha(i)}{\alpha(j)} = G_{i,j}$ for all vertices $1
\leq i,j \leq |V|$. Then $\alpha$ satisfies the product condition $\prod_{v \in
h} \alpha(v) = \pm I^{\otimes n}$ for all hyperedges $h \in H$.
\end{lemma}

\begin{proof}
By the correspondence with symplectic polar spaces (\Cref{symplSec}) it is
equivalent to prove that $\sum_{v \in h} \alpha(v) = 0$ for the null vector $0$
in $\mathbb{F}_2^{2n}$. As already mentioned in the proof
of~\Cref{vecCnsQAprop}, the commutation condition for $\alpha$ entails that $G$
is the Gram matrix of $\alpha$. Let $b = (c_1, \ldots, c_r)$ be composed of $r$
columns of $G$ and forming a basis of the vector space $\text{span}(G)$ spanned
by the columns of $G$. By~\Cref{lemma3tlc} applied to the tuple of vectors
$(\alpha(v))_{v \in V}$, to their Gram matrix $G$ and to the subset of vectors
in $b$, the corresponding vectors $\alpha(v_1)$, \ldots, $\alpha(v_r)$ are
linearly independent vectors in $\mathbb{F}_2^{2n}$. By Lemma~\ref{rankProp}, $r
= 2n = \text{dim}(\mathbb{F}_2^{2n})$, so these vectors form a basis of
$\mathbb{F}_2^{2n}$.

For $1 \leq j \leq r$, we have on the one hand
\begin{align*}
\sum_{v\in h} G_{v,v_j}
= & \sum_{v\in h} \innerproduct{\psi(\alpha(v))}{\psi(\alpha(v_j))} \qquad \text{(by the commutation condition)}
\\
= & \innerproduct{\sum_{v\in h} \psi(\alpha(v))}{\psi(\alpha(v_j))} \qquad \text{(by linearity of the symplectic product)}.
\end{align*}
On the other hand,
$\sum_{v\in h} G_{v,v_j} = 0$
by the assignability condition~(\ref{anticommMatrixColumnHyperedgeZero}). So
\begin{align}
\label{symplProductNull}
\innerproduct{\sum_{v\in h} \alpha(v)}{\alpha(v_j)} & = 0
\end{align}
for all the vectors of the basis $(\alpha(v_j))_{1\leq j \leq r}$, which
is possible only if $\sum_{v\in h} \alpha(v) = 0$.
\end{proof}

\subsection{Algorithmic complexity}
\label{algoComplexitySec}

Multiple methods in~\cite[Theorem 8]{TLC22v2} are suggested for generating
labelings, including searching subgraphs of the graph of the whole symplectic
space, or using backtracking techniques. None of them is polynomial.
However, a recent work~\cite[Theorem 2]{ZPL24} describes a polynomial algorithm
for finding such a labeling, but at the cost of not guaranteeing that the
resulting labeling has the smallest possible number of qubits, since it is equal
to the number of anticommuting pairs in an independent set of vectors of the
Gram matrix. For example, an independent set of the doily being a set of five
vectors, the algorithm will return a labeling with at least five qubits, while
our algorithm returns a labeling with only two qubits.

The algorithm presented in~\Cref{labelingSec} is polynomial, with a complexity
in $O(|V|^3)$, and provides a labeling with the minimal number of qubits. This
is first because the number of iterations of the loop
in~\Cref{alg:findAssignment} is at most $|V|/2$ for a $G$ of full rank. Then,
inside this loop, the assignment of the matrix $B$ on Line~\ref{matrixUpdateB1}
is more costly than the inner loop on Line~\ref{beginFor}, because the matrix
assignment is done in $O(|V|^2)$, while the inner loop is done in $O(|V|)$.

\subsection{Contextuality degree of an assignable hypergram}
\label{abstractCdegDef}

By~\Cref{vecCnsQAprop} all assignable hypergrams admit a Pauli assignment.
By~\Cref{mainResult} all these assignments have the same degree of
contextuality. Putting everything together we propose the following notion of
degree for any assignable hypergram.

\begin{definition}[Contextuality degree of an assignable hypergram]
Let $(V,H,G)$ be an assignable hypergram, $n = \text{rk}(G)/2$ be half the rank
of $G$ (known to be even) and $\alpha$ be the vertex labeling from $V$ to
$\{I,X,Y,Z\}^{\otimes n}~-~\{I^{\otimes n} \}$ computed by the algorithm
\Call{PauliAssignmentFromAnticommutations}{$V,G$}. Let $\text{sgn}_{\alpha} : H
\rightarrow \{-1,1\}$ be its sign function, defined by $\Pi_{v \in h}\,
\alpha(v) = \text{sgn}_{\alpha}(h)\,I^{\otimes n}$ for all hyperedges $h$ in
$H$. The \emph{(contextuality) degree of an assignable hypergram} $(V,H,G)$ is
the minimal Hamming distance between the sign function $\text{sgn}_{\alpha}$ and
the sign function $\text{sgn}_{a}$ of any classical assignment $a : V
\rightarrow \{-1,+1\}$ of its vertices, defined by $\text{sgn}_{a}(h) = \Pi_{v
\in h}~a(v)$ for all hyperedges $h$ in $H$.
\end{definition}

\section{Generalization of former results}
\label{msgdh23Sec}

When the contextuality degree of some $n$-qubit Pauli assignment of some
hypergram is known, the contextuality degree of all $N$-qubit Pauli assignments
of the same hypergram is also known for all the numbers of qubits $N \geq n$,
because~\Cref{mainResult} asserts that all these degrees have the same value.
So, all contextuality degree results established in former work (e.g.,
\cite{MSGDH24,MSGHK24}) only for $n \leq N \leq n'$ for some small number of
qubits $n'$, indeed hold by~\Cref{mainResult} without limit for all $N \geq n$.
Moreover, instead of being obtained as in~\cite{MSGDH24,MSGHK24} after long
computations considering all possible Pauli assignments for all the values of
$N$ in this interval $[n,n']$, they can now be obtained much more efficiently,
by considering only one Pauli assignment of the hypergram they share in common,
with the smallest number of qubits $n$. This section revisits with this larger
point of view several former results about the contextuality degree of
multi-qubit quantum configurations.

Moreover, as announced in~\Cref{ccsSec}, we illustrate here with examples how
the anticommutation relation $G$ added in our framework of hypergrams $(V,H,G)$
to the hypergraph $(V,H)$ of usual quantum configurations opens the door to a
much wider range of cases. To clarify this widening, we classify the quantum
configurations studied in former works~\cite{TLC22v2,DHGMS22,MSGDH24} into two
families. The first family is composed of all the structures whose hypergram
$(V,H,G)$ satisfies $G = \text{cplt}(H)$, where $\text{cplt}(H)$ is the
anticommutation graph of the hypergraph $H$. In other words, in this case, two
observables are in the same context if and only if they commute. The second
family is composed of all the other structures, where $G \subsetneq
\text{cplt}(H)$. In other words, in this case, some commuting pairs of
observables are absent in all contexts.

Each following subsection is devoted to a particular category of quantum
configurations.

\subsection{$1$-spaces}
\label{1spacesSec}

This section is about quantum configurations whose contexts are totally
isotropic subspaces with the projective dimension $1$, also called
\emph{$1$-spaces} or \emph{lines}. Their underlying hypergram belongs to the
first family.

For each number of qubits $n \geq 2$, let $L_{n}$ be the quantum configuration
whose contexts are \textbf{all} the lines of $W_{n}$. For instance, $L_{2}$ is
the 2-qubit doily. The number of observables in $L_{n}$ is $2^{2n} - 1$. Its
number of contexts is
\begin{align}
\frac{(4^{n}-1)(4^{n-1}-1)}{3}.
\end{align}
Its number of negative contexts~\cite{Cab10} is 
\begin{align}
\frac{1}{6} \sum_{c=0}^{n-2}
\sum_{a,b} 3^{2n-a-b-2c}\binom{n}{c} \binom{n-c}{a} \binom{n-c-a}{b}.
\end{align}
For $n \leq 7$, these numbers are respectively given in this order in the
second, third and four column of~\Cref{k1Table}. Its last two columns gathers
the results from Tables 1 and 3 in~\cite{MSGHK24}, up to seven qubits. We
present here neither better bounds for the contextuality degree $d$ nor a
speedup for the computation time of its upper bounds (displayed in the last
column), obtained by the heuristic method presented in~\cite{MSGHK24}, run on a
machine with a 5.4 GHz P-cores and 4.3 GHz E-cores Intel Core i9-13900K
processor. These computations use less than 1.3 Gb out of 64 Gb of RAM. What is
new here is the \textbf{interpretation} of these data, detailed in the following
paragraph.

\begin{table}[hbt!]
\begin{center}
\begin{tabular}{|c|r|r|r|l|l|}
\hline
$n$          & \# obs.      & \# contexts      & \# neg. contexts & Value or bounds for $d$  & Duration\\
\hline
\hline
2            & 15           & 15               & 3                & 3    & $0.01$ s\\
\cline{1-6}
3            & 63           & 315              & 90               & $63$ & $0.1$ s\\
\cline{1-6}
4            & 255          & \np{5355}        & \np{1908}        & $\np{1071} \leq d \leq \np{1575}$ & $0.1$ s\\
\cline{1-6}
5            & \np{1023}    & \np{86955}       & \np{35400}       & $\np{17391} \leq d \leq \np{31479}$ & $1$ s\\
\cline{1-6}
6            & \np{4095}    & \np{1396395}     & \np{615888}      & $\np{279279} \leq d \leq \np{553140}$ & $2$ mn\\
\cline{1-6}
7            & \np{16383}    & \np{22362795}    & \np{10352160}    & $\np{4472559} \leq d \leq \np{9405663}$ & 1 h 34 mn\\
\hline
\end{tabular}
\end{center}
\caption{Dimensions, exact values (for $n = 2,3$) or bounds (for $n \geq 4$) for
the contextuality degree $d$ of quantum configurations in $W_N$ (for all $N \geq
n$) isomorphic to the quantum configuration $L_n$ whose contexts are all the
lines of $W_{n}$.\label{k1Table}}
\end{table}

Each row in~\Cref{k1Table} not only concerns $L_{n}$, but also all the quantum
configurations isomorphic to $L_{n}$ whose contexts are distinct lines of a
symplectic space $W_N$ for some $N \geq n$. Of course, when $N > n$, these
quantum configurations do not contain \textbf{all} the lines of $W_N$, and their
numbers of negative contexts can differ from that of $L_{n}$. However,
\Cref{mainResult} guarantees that all of them have the same contextuality
degree, whose value is either exactly given or bounded in the fifth column
of~\Cref{k1Table}, for $2 \leq n \leq 7$.

It was already known that the contextuality degree of all $n$-qubit doilies is
3, for $n\geq 2$. This was formerly proved by computing this degree for the 12
possible configurations of their negative lines~\cite{MSGDH22}. A direct
consequence of~\Cref{mainResult} is a much simpler proof, which does not rely on
such an enumeration, but justifies only that the contextuality degree of the
2-qubit doily is~3.

In the same way, it is known that the contextuality degree of $L_3$ (all the
3-qubit lines) is 63 and that there is a minimal subset of unsatisfied
hyperedges isomorphic to the split Cayley hexagon of order two~\cite{MSGDH24}.
By~\Cref{mainResult} the contextuality degree of all quantum configurations
isomorphic to $L_3$ (i.\,e., having the same underlying hypergram as it) labeled
by $N$-qubit observables with $N \geq 3$ is also 63. Moreover,
by~\Cref{assignmentCopy}, all these configurations share the same subset of
unsatisfied hyperedges. So, we now know that one of the minimal subsets of
unsatisfied hyperedges in any Pauli assignment of $L_3$ by $N$-qubit observables with
$N \geq 3$ is isomorphic to the split Cayley hexagon of order two.

\subsection{Two-spreads}
\label{2spreadSec}

The case of two-spreads is of interest because all the two-spreads considered up
to now in relation with contextuality~\cite{MSGDH24} are small magic sets whose
underlying hypergram $S_{2s}$ belongs to the second family, as detailed
in~\Cref{ex1}.

For $n \geq 2$, it is known that all $n$-qubit two-spreads are contextual, and
that their contextuality degree is~1~\cite[Proposition~7]{MSGDH24}. The proof of
this proposition in~\cite{MSGDH24} relies on the fact that two-spreads feature
an odd number of negative contexts. The latter fact relies on a careful
inspection of $72 = 6\ctimes 12$ possible configurations of their negative lines,
obtainable by removing one of its 6 spreads of lines from one of the 12 possible
configurations of negative lines in an $n$-qubit doily. Thanks
to~\Cref{mainResult}, we provide here the following much simpler proof of a
similar proposition about two-spreads, more precise by expliciting its
underlying hypergram.

\begin{proposition}
\label{2spreadProp}
The contextuality degree of all $n$-qubit labelings of the two-spread hypergram
$S_{2s}$ is 1.
\end{proposition}

\begin{proof}
As a point-line geometry, disregarding line signs, any $n$-qubit two-spread is
isomorphic to the two-spread of the 2-qubit doily presented in~\Cref{ex1} and
represented in~\Cref{figTwospread}. The latter contains only one negative line,
so its contextuality degree $d$ is at most 1. Moreover, it is a magic set, which
implies that it is contextual~\cite{HS17}, so $d \geq 1$. Consequently, $d =1$,
and by isomorphism and~\Cref{mainResult}, the contextuality degree of all
$n$-qubit two-spreads whose underlying hypergram is $S_{2s}$ is also 1.
\end{proof}

In order to illustrate the impact of the anticommutation graph on contextuality,
the following example presents a non-contextual two-spread embeddable in $W_3$.

\begin{example}
Let us consider the hypergram $S' = (V_{2s},H_{2s},G')$, variant of the
two-spread hypergram $S_{2s} = (V_{2s},H_{2s},G_{2s})$, with the same underlying
hypergraph $(V_{2s},H_{2s})$ but the anticommutation graph $G' = \{\{1,5\}$,
$\{1,7\}$, $\{1,8\}$, $\{1,9\}$, $\{1,12\}$, $\{1,15\}$, $\{2,5\}$, $\{2,8\}$,
$\{2,10\}$, $\{2,11\}$, $\{2,12\}$, $\{3,7\}$, $\{3,9\}$, $\{3,10\}$,
$\{3,11\}$, $\{3,15\}$, $\{4,5\}$, $\{4,7\}$, $\{4,10\}$, $\{4,11\}$,
$\{4,13\}$, $\{4,14\}$, $\{5,11\}$, $\{5,12\}$, $\{6,7\}$, $\{6,8\}$,
$\{6,12\}$, $\{6,13\}$, $\{6,14\}$, $\{7,8\}$, $\{7,10\}$, $\{7,13\}$,
$\{7,15\}$, $\{8,9\}$, $\{8,11\}$, $\{9,11\}$, $\{9,12\}$, $\{9,13\}$,
$\{10,12\}$, $\{10,14\}$, $\{11,14\}$, $\{11,15\}$, $\{12,13\}$, $\{12,14\}$,
$\{14,15\}\}$ different from $G_{2s}$. Whereas $S_{2s}$ is contextual, with
degree 1, this variant $S'$ is not contextual. This hypergram $S'$ and its
3-qubit labeling produced by~\Cref{alg:findAssignment} are presented
in~\Cref{figPseudoTwospread}. This two-spread labeling is a genuine three-qubit
two-spread, which lives in $W_3$ but can be found neither in $W_2$ nor in a
doily of $W_3$.
\end{example}

\begin{figure}[htb!]
  \begin{center}
\begin{tikzpicture}[every plot/.style={smooth, tension=2},
  scale=2.7,
  every node/.style={scale=0.65,circle,draw=black,fill=white,minimum size=1.4cm}
]
  \coordinate (IX) at (0,1.0);
  \coordinate (ZX) at (0,-0.80);
  \coordinate (ZI) at (0,-0.40);
  \coordinate (XX) at (-0.95,0.30);
  \coordinate (ZZ) at (0.76,-0.24);
  \coordinate (YY) at (0.38,-0.12);
  \coordinate (YZ) at (-0.58,-0.80);
  \coordinate (YI) at (0.47,0.65);
  \coordinate (IZ) at (0.23,0.32);
  \coordinate (XY) at (0.58,-0.80);
  \coordinate (XI) at (-0.47,0.65);
  \coordinate (IY) at (-0.23,0.32);
  \coordinate (YX) at (0.95,0.30);
  \coordinate (ZY) at (-0.76,-0.24);
  \coordinate (XZ) at (-0.38,-0.12);
\draw (IX) -- (XI) -- (XX);
\draw (YZ) -- (ZX) -- (XY);
\draw (XY) -- (ZZ) -- (YX);
\draw (YX) -- (YI) -- (IX);
\draw (XX) -- (ZY) -- (YZ);
\draw plot coordinates{(ZI) (ZY) (IY)};
\draw [style=double] plot coordinates{(IZ) (ZZ) (ZI)};
\draw [style=double] plot coordinates{(YY) (ZX) (XZ)};
\draw plot coordinates{(IY) (YI) (YY)};
\draw plot coordinates{(XZ) (XI) (IZ)};
\node at (ZY) {$\vqa{13}{ZXI}$};
\node at (XX) {$\vqa{3}{IXI}$};
\node at (YZ) {$\vqa{14}{ZII}$};
\node at (ZZ) {$\vqa{12}{YXY}$};
\node at (IX) {$\vqa{1}{IIX}$};
\node at (YY) {$\vqa{15}{YYZ}$};
\node at (XI) {$\vqa{2}{IXX}$};
\node at (YX) {$\vqa{11}{YYX}$};
\node at (XZ) {$\vqa{6}{XXI}$};
\node at (ZI) {$\vqa{8}{ZXZ}$};
\node at (IY) {$\vqa{5}{IIZ}$};
\node at (XY) {$\vqa{7}{IZZ}$};
\node at (YI) {$\vqa{10}{YYI}$};
\node at (IZ) {$\vqa{4}{XIX}$};
\node at (ZX) {$\vqa{9}{ZZZ}$};
\end{tikzpicture}
  \end{center}
  \caption{Illustration of the non-contextual hypergram $S' =
  (V_{2s},H_{2s},G')$ sharing the same underlying hypergraph $(V_{2s},H_{2s})$
  as the two-spread hypergram $S_{2s}$, but with a different anticommutation
  graph $G'$. The vertices are represented by the numbers from 1 to 15. The
  hyperedges are represented by the lines. The two negative lines are
  represented by the double lines. The anticommutation graph $G'$ is not shown
  here to keep the figure readable. \label{figPseudoTwospread}}
\end{figure}
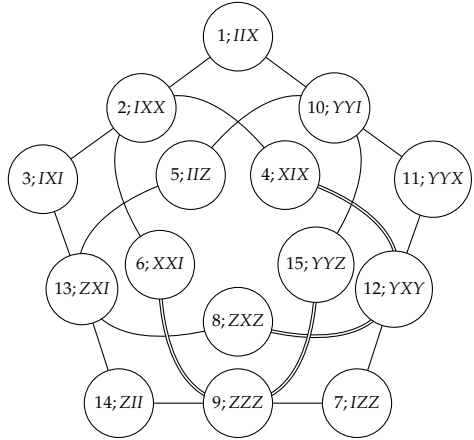

\subsection{Other quantum configurations}
\label{otherSec}

Other quantum configurations studied in former work include $k$-spaces for $k
\geq 2$, Mermin-Peres squares and quadrics. They all belong to the first family.
As for the elliptic and hyperbolic quadrics, they belong to the first family by
definition, because their lines contain all the lines passing through the
corresponding points satisfying their corresponding quadratic form.

Using an optimization algorithm, Saniga, Holweck, Kelleher and we~\cite{MSGHK24}
have found upper bounds for the degree of contextuality of all hyperbolic and
elliptic quadrics of $W_n$, for $4 \leq n \leq 7$. For $n = 4$, this upper bound
is 315 and is the same for hyperbolic and elliptic quadrics. For $n = 5$, the
bounds for elliptic and hyperbolic quadrics respectively are \np{7087} and
\np{6975}. For $n = 6$, they respectively are \np{131700} and \np{132391}. For
$n = 7$, they respectively are \np{2294580} and \np{2331191}. While there are
numerous different hyperbolic and elliptic quadrics in $W_n$, all hyperbolic
(resp. elliptic) quadrics with the same number $n$ of qubits share the same
hypergram. So, by~\Cref{mainResult}, they all have the same degree, and it is
sufficient to estimate it for one of them. This property has been exploited in
our former work~\cite{MSGHK24} to dramatically reduce the computation effort.

\section{Related work}
\label{seCcompare}

This section details in what sense our results can be considered as extensions
or improvements of results from~\cite{TLC22v2}, with an emphasis on the
algorithmic point of view.

We start by recalling some definitions from~\cite{TLC22v2}. An \emph{Eulerian
hypergraph} $\mathcal{H} = (V,H)$ is a hypergraph whose vertices are in an even
number of distinct hyperedges. A \emph{Pauli-based assignment} is defined
in~\cite{TLC22v2} as a \textbf{magic} assignment $\alpha : V \rightarrow
\mathcal{P}^{\otimes n}$ whose values are in the $n$-qubit Pauli group. A magic
assignment satifies the condition $\prod_{v\in h} \alpha(v) = -I$ for an odd
number of hyperedges $h \in H$, hereafter called the \emph{oddness condition}.

Our~\Cref{mainResult} and its proof are similar to Proposition 14
of~\cite{TLC22v2} and its proof, but it holds for all Pauli assignments, our
more general notion than Pauli-based assignments in~\cite{TLC22v2}, free of its
unnecessary oddness condition (odd number of negative hyperedges). It also holds
under one less assumption on the hypergraph, that we do not assume to be
Eulerian, i.\,e., it is not necessarily for each vertex to be in an even number
of distinct hyperedges. With three lines incident to each vertex the doily is a
significant example without this property. Pauli assignments admit two
restrictions not present in~\cite{TLC22v2}: they are injective and assign only
Pauli observables with phase 1. However these restrictions can be considered as
technical details that do not significantly weaken the results: duplicating
labels would have no interest, there are already plenty of interesting
phase-free assignments to study and all examples in~\cite{TLC22v2} and other
work only consider vertices labeled by phase-free Pauli observables.

Finally, our claim that the labeling $\alpha$ returned by
\Call{PauliAssignmentFromAnticommutations}{$V,G$} is a Pauli assignment and its
justification (in \Cref{justifSec}) are a generalization of Proposition~5
in~\cite{TLC22v2} and its proof, because here the hypergraph $(V,H)$ is again
not assumed to be Eulerian. Moreover, from the algorithmic point of view, our
polynomial~\Cref{alg:findAssignment} is more efficient than 
the two approaches mentioned in~\cite{TLC22v2}, the first one being a
search for subgraphs of the graph of the whole symplectic space, and the second
one using backtracking techniques. Both of them have an exponential complexity:
first in the number of vertices because of the nature of the algorithms used,
and second in the number of qubits because of the exponentially growing size 
of the symplectic space $W_n$.

\section{Conclusion}

The notion of assignable hypergram proposed in this paper can not only be
attached an abstract notion of contextuality degree, as detailed here, but
characterizes an exploration space where original state-independent
Kochen-Specker proofs can be looked for. Finding efficient ways to explore that
space is the main perspective. A preliminary perspective is to add criteria to
reduce the size of the space and orient the search.

The proposed framework includes the well-known magic sets, but is much wider.
Magic sets are attractive notably because their contextuality can be proved by a
simple human reasoning. However they have restrictions that we show here to be
unnecessary for the existence of Pauli assignments, such as the oddness
condition. When considering more general objects, we accept to loose the nice
property of a simple proof of contextuality and to rely on software to decide
contextuality and to compute bounds for the contextuality degree, as
in~\cite{MSGDH24}.

\section*{Acknowledgments} We thank Stefan Trandafir for noticing us the link
between our research area and~\cite{TLC22,TLC22v2}, which led us to the results
of this article. We also thank the anonymous reviewer, Frédéric
Holweck, Isabelle Jacques, Colm Kelleher, Pierre-Alain Masson and Metod Saniga
for their suggestions to improve this text.

\section*{Financial support} This work is supported by the PEPR integrated
project EPiQ ANR-22-PETQ-0007 part of Plan France 2030, by the project TACTICQ
of the EIPHI Graduate School (contract ANR-17-EURE-0002) and by the
Bourgogne-Franche-Comté Region.

\section*{Author declarations}

\paragraph*{Conﬂict of Interest.} The authors have no conﬂicts to
disclose.

\paragraph*{Author contributions.} Axel Muller: Conceptualization; Software
(main developer); Writing (50\%). Alain Giorgetti: Methodology; Software
(advice); Supervision; Writing (50\%).

\section*{Data Availability Statement} The data that support the findings of
this study are openly available in
\url{https://github.com/quantcert/quantcert.github.io/releases/tag/v1.2.0}.

\printbibliography

\end{document}